%% file: fm18.tex
\documentclass[pdftex,runningheads]{llncs} 
\usepackage{microtype}
\usepackage[latin1]{inputenc}
\usepackage{amsmath,amsfonts,amssymb,latexsym,colonequals}
\usepackage{stmaryrd,proof}
\usepackage{hyperref}
\usepackage[inline]{enumitem}
\usepackage[boxed]{algorithm}
\usepackage[noend]{algorithmic}
\usepackage{siunitx} 
\usepackage{todonotes}
\usepackage{multirow}
\usepackage{subcaption}
\usepackage{tikz}
\usetikzlibrary{calc} 
\usetikzlibrary{shapes.multipart,arrows}
\usetikzlibrary{patterns}
\usepackage{wrapfig}
\usepackage{thm-restate}
\usepackage{macros-fm18}

\title{Optimal and Robust Controller Synthesis%
  \thanks{Work supported by ERC projects Lasso and EQualIS.}}
\subtitle{Using Energy Timed Automata with Uncertainty}

\author{Giovanni Bacci\inst{1} \and Patricia Bouyer\inst{2} \and Uli Fahrenberg\inst{3}%
\and Kim G.~Larsen\inst{1} \and \\ Nicolas Markey\inst{4} \and Pierre-Alain Reynier\inst{5}}
\institute{
Department of Computer Science, Aalborg University, Denmark 
\and
LSV, CNRS \& ENS Cachan, Universit\'e Paris-Saclay, Cachan, France 
\and
\'Ecole Polytechnique, Palaiseau, France 
\and
Univ. Rennes, IRISA, CNRS \& INRIA, Rennes, France 
\and
Aix-Marseille Univ, LIF, CNRS, Marseille, France 
}

\titlerunning{Optimal and Robust Controller Synthesis}
\authorrunning{G.~Bacci, P.~Bouyer, U.~Fahrenberg, K.G.~Larsen, N.~Markey, P.A.~Reynier}

\begin{document}

\maketitle

\begin{abstract}
  In this paper, we propose a novel framework for the synthesis of
  robust and optimal energy-aware controllers.  The~framework is based
  on energy timed automata, allowing for easy expression of
  timing constraints and variable energy rates.  We prove decidability
  of the energy-constrained infinite-run problem in settings with both
  certainty and uncertainty of the energy rates. We also consider the
  optimization problem of identifying the minimal upper bound that
  will permit existence of energy-constrained infinite runs.
  Our algorithms are based on quantifier elimination for linear real
  arithmetic.  Using Mathematica and Mjollnir, we illustrate our
  framework through a real industrial example of a hydraulic oil pump.
  Compared with previous approaches our method is completely automated
  and provides improved results.
\end{abstract}



\section{Introduction}

\input{introduction}

\section{Energy Timed Automata} \label{sec:ETA}

\input{eta}

\section{Energy Timed Automata with Uncertainties} \label{sec:ETAu}

\input{etau}

\section{Case Study}
\label{sec-hydac}

\input{hydac-short}

\section{Conclusion}
\input{conclu}

\bibliographystyle{abbrv}
\bibliography{bibfm18}

\clearpage
\appendix
\input{appendix}

\end{document}

%% file: introduction.tex
Design of controllers for embedded  systems is a difficult engineering
task. Controllers must  ensure a variety of safety  properties as well
as optimality with respect to given performance properties.  Also, for
several                          systems,                         e.g.
\cite{DBLP:conf/fm/BisgaardGHKNS16,DBLP:journals/fac/BochmannHLO17,DBLP:conf/birthday/PhanHM14},
the properties involve non-functional aspects such as time and energy.

\looseness=-1
We  provide a  novel framework  for  automatic synthesis  of safe  and
optimal controllers  for resource-aware systems based  on \emph{energy
  timed automata}.   Synthesis of  controllers is obtained  by solving
time-  and energy-constrained  infinite  run  problems.  Energy  timed
automata \cite{BFLMS08}  extend timed automata
\cite{AD94} with a continuous \emph{energy} variable that evolves with
varying rates and discrete updates  during the behaviour of the model.
Closing an open problem from \cite{BFLMS08}, we
prove  decidability of  the infinite  run problem  in settings,  where
rates  and updates  may be  both  positive and  negative and  possibly
subject to  uncertainty.  Additionally, the accumulated  energy may be
subject to lower and upper  bounds reflecting constraints on capacity.
Also  we consider  the  optimization problems  of identifying  minimal
upper   bounds   that   will   permit  the   existence   of   infinite
energy-constrained runs.  Our decision and optimization algorithms for
the energy-constrained infinite run problems are based on reductions to
quantifier  elimination  (QE) for  linear  real arithmetic,  for  which  we
combine Mathematica \cite{Mathematica}  and Mjollnir \cite{Monniaux10}
into a tool chain.

\definecolor{arylideyellow}{rgb}{0.91, 0.84, 0.42}
\definecolor{azure}{rgb}{0.94, 1.0, 1.0}
\begin{figure}[t]
\centering
\begin{subfigure}[b]{0.4\textwidth}
\begin{tikzpicture}[scale=1,
  component/.style={rectangle, rounded corners, draw=gray!90, thick}] 
  \draw 
  (0,0) node[component] (pump) {\parbox{10ex}{\centering Pump}}
  ($(pump)+(down:1.5)$) node[component] (machine) {\parbox{10ex}{\centering Machine}}
  ;
  
  \path[-latex,double, font=\scriptsize]
  	(pump) edge[thick,double] node[above] {$2.2$ \si{\litre/\second}} ($(pump)+(right:2)$)
	($(machine)+(right:2)$) edge[thick,double] (machine)
  	;
	
  \draw[fill=arylideyellow] ($(pump.north)+(right:2)$) rectangle ($(machine.south)+(right:3.5)$);
  \draw[fill=azure] ($(pump.north)+(right:2)$) rectangle ($(machine.south)+(3.5,1.3)$);
  \draw[color=red,dashed] ($(machine.south)+(2,1.5)$) -- ($(machine.south)+(3.6,1.5)$) 
  	($(machine.south)+(4,1.5)$) node {\small{$V_{max}$}};
  \draw[color=red,dashed] ($(machine.south)+(2,0.5)$) -- ($(machine.south)+(3.6,0.5)$)
  	($(machine.south)+(4,0.5)$) node {\small{$V_{min}$}} ;
   \draw ($(pump.north)+(2.75,0.3)$) node {Accumulator};
   
   \draw[color=black,very thin] (-0.9,1) rectangle (4.4,-2.2);
\end{tikzpicture}
\caption{System Components} \label{fig:hydacoverview}
\end{subfigure}
\qquad
\begin{subfigure}[b]{0.5\textwidth}
  \includegraphics[width=1 \textwidth]{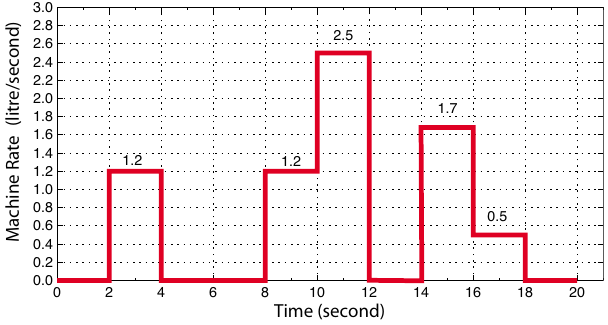}%
  \caption{Cycle of the Machine}
  \label{fig:machinecycle}
 \end{subfigure}
 \caption{Overview of the HYDAC system}
 \label{fig-hydacexpl}
\end{figure}
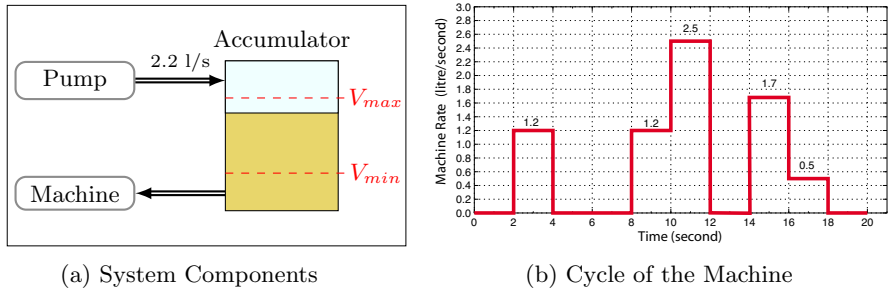
To demonstrate the applicability of our framework, we revisit an
industrial case study provided by the HYDAC company in the context of
the European project
Quasimodo%
~\cite{quasimodo}.
It~consists in an on/off control system
(see~Fig.~\ref{fig:hydacoverview})  composed of
\begin{enumerate*}[label=(\roman*)]\item a machine that consumes oil according to a cyclic pattern of 20~s~(see~Fig.~\ref{fig:machinecycle}), 
\item an accumulator containing oil and a fixed amount of gas in order
  to put the  oil under pressure, and \item a  controllable pump which
  can pump oil into the accumulator with rate 2.2 l/s.
\end{enumerate*}
%
The control  objective for switching the  pump on and off  is twofold:
first the level of oil in  the accumulator (and so the gas pressure)
shall be  maintained within  a safe  interval; second,  the controller
should try  to minimize the  (maximum and  average) level of  oil such
that the pressure in the system  is kept minimal.
We show how to model this system,
with varying constraints on pump  operation, as energy timed automata.
Thus our tool  chain may automatically synthesize  guaranteed safe and
optimal control strategies.

The HYDAC case was first considered in~\cite{CJLRR09} as a timed game
using the tool
\textsc{Uppaal-Tiga}~\cite{DBLP:conf/concur/CassezDFLL05,DBLP:conf/cav/BehrmannCDFLL07}
for synthesis.  Discretization of oil-level (and time) was used to
make synthesis feasible.  Besides limiting the opportunity of
optimality, the discretization also necessitated posterior
verification using PHAVER~\cite{phaver} to rule out possible resulting
incorrectness.  Also, identification of safety and minimal oil levels
were done by manual and laborious search.
In~\cite{DBLP:journals/tcst/MiremadiFAL15} the timed game models
of~\cite{CJLRR09} (rephrased as Timed Discrete Event Systems) are
reused, but BDDs are applied for compact representation of the
discrete oil-levels and time-points encountered during synthesis.
%
\cite{DBLP:conf/emsoft/JhaST11}  provides  a  framework  for  learning
optimal  switching  strategies  by   a  combination  of  off-the-shelf
numerical optimization and generalization  by learning.  The~HYDAC case
is  one  of the  considered  cases.   The~method offers  no  absolute
guarantees of hard constraints on energy-level, but rather attempts to
enforce these through the use of high penalties.
%
\cite{DBLP:conf/fm/ZhaoZKL12}  focuses exclusively  on the  HYDAC case
using a direct  encoding of the safety-  and optimality-constraints as
QE problems.  This  gives---like in our case---absolute guarantees.
However,  we~are  additionally  offering  a  complete  and  decidable
framework based  on energy  timed automata,  which extends  to several
other systems.
Moreover,  the controllers  we obtain  perform significantly  better
than   those  of   \cite{CJLRR09}  and   \cite{DBLP:conf/fm/ZhaoZKL12}
(respectively  up   to  22\%  and   16\%  better)  and   are  obtained
automatically by  our tool  chain combining Mjollnir  and Mathematica.
This   combination   permits   quantifier   elimination   and   formula
simplification  to be  done  in a  compositional  manner, resulting  in
performance surpassing  each tool individually.  We~believe  that this
shows  that our  framework  has a  level of  maturity  that meets  the
complexity of several relevant industrial control problems.


Our work is related to controllability of (constrained) piecewise affine~(PWA)~\cite{BemporadFTM00} and hybrid systems~\cite{AlurCHH93}. In particular, the energy-constrained infinite-run problem is related to the so called \emph{stability problem} for PWAs. Blondel and Tsitsiklis~\cite{BlondelT99} have shown that verifying stability of autonomous piecewise-linear~(PWL) systems is \NP-hard, even in the simple case of two-component subsystems; several global properties (e.g.~global convergence, asymptotic stability and mortality) of PWA systems have been shown undecidable in~\cite{BlondelBKT01}.

%% file: eta.tex
\label{sec-eta}

\subsubsection{Definitions.}



%




Given a finite set~$\Cl$ of clocks, the set of \emph{closed clock
  constraints} over~$\Cl$, denoted~$\Constr(\Cl)$, is the set of
formulas built using
\(
g \coloncolonequals x \sim n \mid g\wedge g
\),
where $x$ ranges over~$\Cl$, $\mathord\sim$ ranges over
$\{\mathord\leq,\mathord\geq\}$ and $n$ ranges over~$\bbQ+$. That
a~clock valuation~$v\colon \Cl \to \bbR+$ satisfies a clock
constraint~$g$, denoted~$v\models g$, is defined in the natural way.
For~a~clock valuation~$v$, a~real~$t\in\bbR+$, and a
subset~$R\subseteq \Cl$, we~write $v+t$ for the valuation mapping each
clock~$x\in\Cl$ to~$v(x)+t$, and $v[R\to 0]$ for the valuation mapping
clocks in~$R$ to~zero and clocks not in~$R$ to their value
in~$v$. Finally we write $\mathbf{0}_X$ (or simply $\mathbf{0}$) for
the clock valuation assigning $0$ to every $x \in X$.

For $E\subseteq\bbR$, we~let $\calI(E)$ be the set of closed intervals
of~$\bbR$ with bounds in~$E\cap \bbQ$. Notice that any interval
in~$\calI(E)$ is bounded, for any~$E\subseteq\bbR$.

\begin{definition}
  An \emph{energy timed automaton} (\emph{ETA} for short;
  a.k.a. \emph{priced} or \emph{weighted timed
    automaton}~\cite{ALP01,BFH+01}) is a
  tuple~$\calA=\tuple{S,S_0,\Cl,\inv,\rate,T}$ where $S$ is a finite
  set of states, $S_0\subseteq S$ is the set of initial states,
  $\Cl$~is a finite set of clocks, $\inv\colon S \to \Constr(\Cl)$
  assigns invariants to states, $\rate\colon S \to \bbQ$ assigns rates
  to states, and $T\subseteq S\times \Constr(\Cl)\times \bbQ \times
  2^{\Cl}\times S$ is a finite set of transitions.

  An~\emph{energy timed path} (\emph{ETP}, a.k.a. \emph{linear energy
    timed automaton}) is an energy timed automaton for which $S$ can
  be written as $\{s_i \mid 0\leq i\leq n\}$ in such a way that
  $S_0=\{s_0\}$, and $T=\{(s_i,g_i,u_i,z_i,s_{i+1}) \mid 0\leq i<n\}$.
  We~additionally require that all clocks are reset on the last
  transition, i.e.,~$z_{n-1}=\Cl$.
\end{definition}

Let $\calA=\tuple{S,S_0,\Cl,\inv,\rate,T}$ be an ETA.
A~\emph{configuration} of $\calA$ is a triple $(\ell,v,w) \in S \times
(\bbR+)^\Cl \times \bbR$, where $v$ is a clock valuation, and $w$ is
the energy level.  Let $\tau = (t_i)_{0 \le i < n}$ be a finite
sequence of transitions, with $t_i = (s_i,g_i,u_i,z_i,s_{i+1})$ for
every $i$.  A~finite \emph{run} in $\calA$ on $\tau$ is a sequence of
configurations $\rho=(\ell_j,v_j,w_j)_{0\leq j\leq 2n}$ such that
there exists a sequence of delays~$(d_i)_{0\leq i<n}$ for which the
following requirements hold:
\begin{itemize}
\item for all $0\leq j<n$, $\ell_{2j}=\ell_{2j+1}=s_j$, and
  $\ell_{2n}=s_n$;
\item for all $0\leq j<n$, $v_{2j+1}=v_{2j}+d_j$ and
  $v_{2j+2}=v_{2j+1}[z_j\to 0]$;
\item for all $0\leq j<n$, $v_{2j}\models \inv(s_j)$ and
  $v_{2j+1}\models \inv(s_j) \wedge g_j$;
\item for all $0\leq j<n$, 
  $w_{2j+1}=w_{2j} + d_j\cdot r(s_j)$ and
  $w_{2j+2}=w_{2j+1} + u_j$.
\end{itemize}
We will by extension speak of runs read on ETPs (those runs will then
end with clock valuation $\mathbf{0}$).
The notion of infinite run is defined similarly.  Given~$E\in
\calI(\bbQ)$, such a run is said to satisfy energy constraint~$E$ if
$w_j\in E$ for all $j$.

\begin{example}
Fig.~\ref{fig-extp} displays an example of an ETP~$\calP$ and one of
its runs~$\rho$. Since no time will be spent in~$s_2$, we~did not
indicate the invariant and rate of that state.
%
%
The~sequence~$\rho$ is a run of~$\calP$. Spending $0.6$ time units
in~$s_0$, the value of clock~$x$ reaches~$0.6$, and the energy level
grows to~${3+0.6\times 2=4.2}$; it~equals ${4.2 - 3=1.2}$ when
entering~$s_1$. Then $\rho$~satisfies
energy constraint~$[0;5]$.
\end{example}

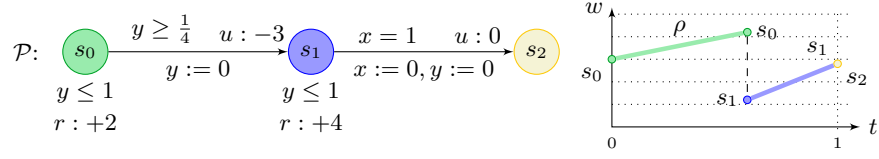
\begin{figure}[t]
  \centering
  \begin{tikzpicture}
    \begin{scope}[xscale=.75]
    \draw (0,0) node[rond6,vert] (a) {$s_0$}
      node[below=3mm] {$\stack{y\leq 1}{\rate:+2}$} node[left=5mm] {$\calP$:};
    \draw (4,0) node[rond6,bleu] (b) {$s_1$}
      node[below=3mm] {$\stack{y\leq 1}{\rate:+4}$};
    \draw (8,0) node[rond6,jaune,opacity=1] (c) {$s_2$};
    \draw (a) edge[-latex'] node[pos=.3,above] {$y\geq \frac14$}
      node[pos=.8,above] {$\update:-3$}
      node[below] {$y:=0$} (b);
    \draw (b) edge[-latex'] node[pos=.3,above] {$x=1$} node[below] {$x:=0,y:=0$}
      node[above, pos=.8] {$\update: 0$}  (c);
    \end{scope}

    \begin{scope}[yshift=-1cm,xshift=7cm,xscale=5,yscale=.5,scale=.6]
      \draw[latex'-latex'] (0,5.3) node[left] {$w$} |- (1.1,0) node[right] {$t$};
      \foreach \i in {1,2,3,4,5}
        {\draw[dotted] (0,\i) -- +(1.05,0);}
      \draw[dotted] (1,-.1) node[below] {$\scriptstyle 1$} -- +(0,5.3);
      \path (0,-.1)  node[below] {$\scriptstyle 0$};
      \draw (0,3) node[rond1,vert] (a) {} node[below left] {$s_0$};  
      \draw (0.6,4.2) node[rond1,vert] (b) {} node[right] {$s_0$};
      \draw (0.6,1.2) node[rond1,bleu] (c) {} node[left] {$s_1$};
      \draw (1,2.8) node[rond1,jaune] (d) {} node[above left] {$s_1$} node[below right] {$s_2$};
      \draw[line width=.6mm,fvert] (a) -- (b) node[midway,above,black] {$\rho$};
      \draw[dashed] (b) -- (c);
      \draw[line width=.6mm,fbleu] (c) -- (d);
    \end{scope}
  \end{tikzpicture}
  \caption{An energy timed path~$\calP$, and a run~$\rho$
    of~$\calP$ with initial energy level~$3$.}
  \label{fig-extp}
\end{figure}


\begin{definition}
  A \emph{segmented energy timed automaton} (\emph{SETA} for short) is
  a tuple $\calA=\tuple{S,T,P}$ where $\tuple{S,T}$ is a finite graph
  (whose states and transitions are called \emph{macro-states} and
  \emph{macro-transitions}), $S_0$ is a set of initial macro-states,
  and $P$ associates with each macro-transition $t=(s,s')$ of~$\calA$
  an ETP with initial state~$s$ and~final state~$s'$. We~require that
  for any two different transitions~$t$ and~$t'$ of~$\calA$, the state
  spaces of~$P(t)$ and~$P(t')$ are disjoint and contain no
  macro-states, except (for~both~conditions) for their first and last
  states.
  
  A~SETA is \emph{flat} if the underlying graph~$\tuple{S,T}$~is
  (i.e., for any~$s\in S$, there is at most one non-empty path in the
  graph~$\tuple{S,T}$ from~$s$
  to~itself~\cite{CJ98,BIL06}).
  It~is called
    \emph{depth-1} whenever the
  graph~$\tuple{S,T}$ is tree-like, with only loops at
    leaves.
\end{definition}

A~(finite or infinite) execution of a SETA is a (finite or infinite)
sequence of runs~$\rho=(\rho^i)_i$ such that for all~$i$, writing
$\rho^i=(\ell^i_j, v^i_j, w^i_j)_{0\leq j\leq 2n_i}$,
it~holds:
\begin{itemize}
\item $\ell^i_0$ and $\ell^i_{2n_i}$ are macro-states of~$\calA$, and
  $\rho^i$ is a run of the ETP $P(\ell^i_0, \ell^i_{2n_i})$;
\item $\ell^{i+1}_0=\ell^i_{2n_i}$ and $w^{i+1}_0=w^i_{2n_i}$.
\end{itemize}
Hence a run in a SETA should be seen as the concatenation of
paths~$\rho^i$ between macro-states.
Notice also that each $\rho^i$ starts and ends with all clock values~zero, since all
clocks are reset at the end of each ETP, when a main state is entered. Finally, given an interval $E\in \calI(\bbQ)$, an~execution~$(\rho^i)_i$ satisfies energy constraint~$E$ whenever all
individual runs~$\rho^i$ do.

\begin{remark}
In contrast with ETAs, the class of SETAs is not closed under parallel composition. Intuitively, the ETA resulting from the parallel composition of two SETAs may not be ``segmented'' into a graph of energy timed-paths because the requirement that all clocks are reset on the last transition may not be satisfied. 
Furthermore, parallel composition does not preserve flatness because it may introduce nested loops.
\end{remark}

\begin{example}
  Figure~\ref{fig-exeta} displays a SETA~$\calA$ with two
  macro-states~$s_0$ and~$s_2$, and two macro-transitions.
  The~macro-self-loop on~$s_2$ is associated with the energy timed
  path of Fig.~\ref{fig-extp}.
  The~execution $\rho=\rho^1\cdot (\rho^2\cdot \rho^3)^\omega$ is an
  ultimately-periodic execution of~$\calA$. This~infinite execution
  satisfies the energy constraint $E=[0;5]$ (as well as the (tight)
  energy constraint~$[1;4.6]$).
\end{example}

\begin{figure}[tb]
  \centering
  \begin{tikzpicture}[location/.style={circle, draw=gray!90, thick}]
    \begin{scope}[scale=1,yshift=.6cm]
      \draw 
      (0,0) node[location] (b) {$s_0$}
      node[font=\small,left=8mm] {$(S,T)\colon{}$}
      ($(b)+(1.5,0)$) node[location] (c) {$s_2$};
      \path[-latex, font=\small]
      (b) edge (c)
      (c) edge[loop right] (c)
      (b)+(-.8,0) edge (b);
    
    \path (3,0) node {\begin{tabular}{l} $P_{0,2}=$ \\[1cm]
        $P_{2,2} =$ \end{tabular}};
    \begin{scope}[xshift=4cm,yshift=.7cm]
      \draw (0,0) node[rond6,rouge] (b) {$s_0$} node[below=3mm] {$\scriptstyle r:0$} ;
      \draw (4.5,0) node[rond6,rouge] (c) {$s_2$};
      \draw (2.25, 0) node[rond6,jaune] (c1) {$s_1$} node[below=3mm] {$\scriptstyle r:-1$};
      \draw (b) edge[-latex'] node[below] {$\scriptstyle y:=0$} node[above,pos=.75] {$\scriptstyle u:+1$} (c1);
      \draw (c1) edge[-latex'] node[below] {$\smstack{x:=0}{y:=0}$} node[above] {$\scriptstyle x\geq 1$} (c);
    \end{scope}
    \begin{scope}[xshift=4cm,yshift=-.5cm]
      \draw (0,0) node[rond6,rouge] (c) {$s_2$} node[below=3mm] {$\scriptstyle r:+2$};
      \draw (2.25, 0) node[rond6,orange] (d1) {$s_3$} node[below=3mm] {$\scriptstyle r:+4$};
      \draw (4.5,0) node[rond6,rouge] (cbis) {$s_2$};
      \draw (c) edge[-latex'] node[above,pos=.25] {$\scriptstyle y\geq 0.25$}  node[above,pos=.8] {$\scriptstyle u:-3$}  node[below] {$\scriptstyle y:=0$} (d1);
      \draw (d1) edge[-latex'] node[above] {$\scriptstyle x=1$} node[below] {$\smstack{x:=0}{y:=0}$} (cbis);      
    \end{scope}
    \end{scope}

    \begin{scope}[yshift=-3cm,xshift=-1cm,xscale=3,yscale=.5,scale=.9]
      \draw[latex'-latex'] (0,5.3) node[left] {$w$} |- (3.1,0) node[right] {$t$};
      \foreach \i in {1,2,3,4,5}
        {\draw[dotted] (0,\i) -- +(3.05,0);}
      \foreach \i in {1,2,3}
        {\draw[dotted] (\i,-.1) node[below] {$\scriptstyle \i$} -- +(0,5.3);}
      \path (0,-.1)  node[below] {$\scriptstyle 0$};
      \draw (0,3) node[rond1,rouge] (a) {} node[below left] {$s_0$};  
      \draw (0.2,3) node[rond1,rouge] (b) {} node[right] {$s_0$};
      \draw (0.2,4) node[rond1,jaune] (c) {} node[left] {$s_1$};
      \draw (1,3.2) node[rond1,rouge] (d) {} node[above left] {$s_1$} node[below right] {$s_2$};
      \draw[line width=.6mm,frouge] (a) -- (b);
      \draw[dashed] (b) -- (c);
      \draw[line width=.6mm,djaune] (c) -- (d);
      \path (.5,0.5) node {$\rho^1$};
      
      \draw (1,3.2) node[rond1,rouge] (a') {};
      \draw (1.4,4) node[rond1,rouge] (b') {} node[right] {$s_2$};
      \draw (1.4,1) node[rond1,orange] (c') {} node[left] {$s_3$};
      \draw (2,3.4) node[rond1,rouge] (d') {} node[above left] {$s_3$} node[below right] {$s_2$};
      \draw[line width=.6mm,frouge] (a') -- (b');
      \draw[dashed] (b') -- (c');
      \draw[line width=.6mm,dorange] (c') -- (d');
      \path (1.5,0.5) node {$\rho^2$};

      \draw (2,3.4) node[rond1,rouge] (a'') {};
      \draw (2.6,4.6) node[rond1,rouge] (b'') {} node[right] {$s_2$};
      \draw (2.6,1.6) node[rond1,jaune] (c'') {} node[left] {$s_3$};
      \draw (3,3.2) node[rond1,rouge] (d'') {} node[above left] {$s_3$} node[below right] {$s_2$};
      \draw[line width=.6mm,frouge] (a'') -- (b'');
      \draw[dashed] (b'') -- (c'');
      \draw[line width=.6mm,dorange] (c'') -- (d'');
      \path (2.5,0.5) node {$\rho^3$};
    \end{scope}
    
  \end{tikzpicture}

  \caption{A SETA~$\calA = \protect\tuple{S,T,P}$ with
    implicit global invariant~$y\leq 1$; omitted discrete updates are
    assumed to be zero. The map $P$ associates with each $(s_i,s_j) \in
    T$ the ETP $P_{i,j}$. The~infinite sequence $\rho^1\cdot
    (\rho^2\cdot \rho^3)^\omega$ is an infinite execution of~$\calA$
    with initial energy level~$3$ satisfying the energy
    constraint~$E=[0;5]$.}
  \label{fig-exeta}
\end{figure}
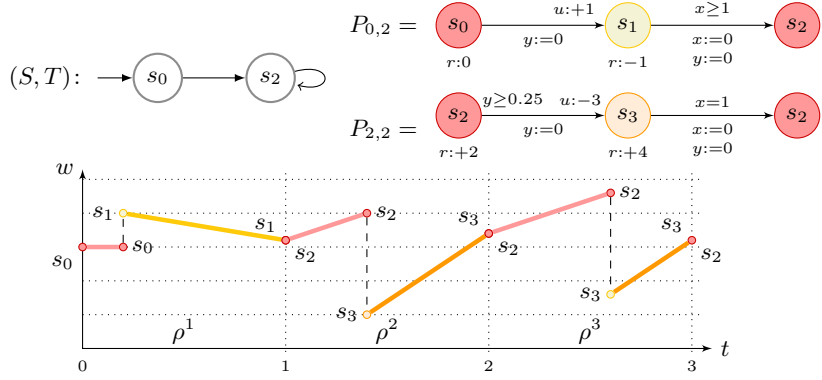

In this paper, we consider the following \emph{energy-constrained
  infinite-run problem}~\cite{BFLMS08}: given an energy timed
automaton~$\calA$ and a designated state~$s_0$, an energy
constraint~$E\in\calI(\bbQ)$ and an initial energy level~$w_0\in E$,
does there exist an infinite execution in~$\calA$ starting
from~$(s_0,\mathbf{0},w_0)$ that satisfies $E$?

In the general case, the energy-constrained infinite-run problem is
undecidable, even when considering ETA with only two
clocks~\cite{Mar11}.
In this paper, we~prove:
\begin{restatable}{theorem}{thmecir}\label{thm-ecir}
The energy-constrained infinite-run problem is decidable for flat
SETA.
\end{restatable}

\begin{restatable}{theorem}{thmoptU}
\label{thm-optU}
  Given a fixed lower bound~$L$, the existence of an upper bound~$U$,
  such that there is a solution to the energy-constrained infinite-run
  problem for energy constraint $E = [L;U]$, is decidable for flat
  SETA. If~such~a~$U$ exists, then for depth-1 flat SETA, we can
  compute the least one.
\end{restatable}
\noindent We~only sketch a proof of the former result, and
refer to~\cite{BacciBFLMR18long} for the full proof.

\subsubsection{Binary energy relations.}
\label{page-subsec:binary_energy_relations}

Let $\calP=\tuple{\{s_i \mid 0 \le i \le n\},\{s_0\},X,I,r,T}$ be an
ETP from $s_0$ to~$s_n$.
Let~$E\subseteq \calI(\bbQ)$ be an energy constraint. The~\emph{binary
  energy relation}~$\calR_\calP^{E} \subseteq E \times E$ for~$\calP$
under energy constraint~$E$ relates all pairs~$(w_0,w_1)$ for which
there is a finite run of~$\calP$ from~$(s_0,\mathbf{0},w_0)$
to~$(s_n,\mathbf{0},w_1)$ satisfying energy constraint~$E$.
This~relation is characterized by the following first-order
formula:
\begin{xalignat*}1
  \calR_\calP^E(w_0,w_1) \iff \exists (d_i)_{0\leq i<n} .\
  & \Phi_{\text{timing}} \wedge \Phi_{\text{energy}} \wedge
    w_1=w_0+\sum_{k=0}^{n-1}(d_k\cdot r(s_k) + u_k) 
\end{xalignat*}
where $\Phi_{\text{timing}}$ encodes all the timing constraints that
the sequence~$(d_i)_{0\leq i<n}$ has to fulfill (derived from guards
and invariants, by expressing the values of the clocks in terms
of~$(d_i)_{0\leq i<n}$), while $\Phi_{\text{energy}}$ encodes the
energy constraints (in~each state, the~accumulated energy must be
in~$E$).

\begin{wrapfigure}[10]{r}{3.6cm}
\centering
  \begin{tikzpicture}
    \begin{scope}[xshift=5cm,yshift=0cm,scale=.5]
      \path[use as bounding box] (-.2,-.2) -- (5.8,5.8);
      \draw[latex'-latex'] (5.8,0) node[right] {$w_0$}
        -| (0,5.8) node[left] {$w_1$};
      \foreach \i in {0,1,...,5}
      {\draw[dotted] (\i,0) -- +(0,5.6);
      \draw (\i,-0.5) node {\scriptsize{\i}};
      \draw[dotted] (0,\i) -- +(5.6,0);
      \draw (-0.5,\i) node {\scriptsize{\i}}; }
      \begin{scope}
        \path[clip] (0,0) -| (5.6,5.6) -| cycle;
        \fill[black!20!white] (2.5,3) -- (4.5,5) -- (3,2) -- (1,0) -- cycle;
        \draw[dashed] (2,0) -- (5,6);
        \draw[dashed] (1,0) -- (4,6);
        \draw[dashed] (1,0) -- (6,5);
        \draw[dashed] (0,.5) -- (6,6.5);        
      \end{scope}
      \begin{scope}
        \path[clip] (5,0) -| (5.6,5.6) -| (0,5) -| cycle;
        \fill[pattern=north west lines,opacity=.25] (0,0) -| (6,6) -| cycle;
      \end{scope}
    \end{scope}
  \end{tikzpicture}
\end{wrapfigure}
It is easily shown that $\calR_P^E$ is a closed, convex subset of
$E\times E$ (remember that we consider closed clock constraints); thus
it~can be described as a conjunction of a finite set of linear
constraints over $w_0$ and~$w_1$ (with non-strict inequalities), using
quantifier elimination of variables~$(d_i)_{0\leq i<n}$.

\begin{example}\label{ex-enerel}
  We illustrate this computation on the ETP of
  Fig.~\ref{fig-extp}. For energy constraint~$[0;5]$, the energy
  relation (after removing redundant constraints) reads as
\begin{xalignat*}1
  \calR_\calP^E(w_0,w_1) \iff \exists d_0 ,d_1.\ &
  d_0\in[0.25;1] \wedge d_1\in [0;1] \wedge d_0+d_1=1 \wedge {} \\
  & w_0\in[0;5] \wedge 
    w_0+2d_0\in [0;5] \wedge w_0+2d_0-3\in[0;5] \wedge{}\\
  &  w_1=w_0+2d_0+4d_1-3 \wedge w_1\in[0;5].
\end{xalignat*}
This simplifies to $(w_1+2 \leq 2w_0\leq w_1+4) \wedge (w_1-0.5\leq
 w_0\leq w_1+1)$. 
The~corresponding polyhedron is depicted above.
\end{example}

\subsubsection{Energy functions.}
We~now focus on properties of energy relations. First notice that
for any interval~$E\in\calI(\bbQ)$, the partially-ordered set
$(\calI(E),\supseteq)$ is $\omega$-complete, meaning that for any chain
$(I_j)_{j\in\bbN}$, with $I_j\supseteq I_{j+1}$ for all~$j$, the limit
$\bigcap_{j\in\bbN} I_j$ also belongs to~$\calI(E)$.
By~Cantor's Intersection Theorem, if additionally each interval~$I_j$
is non-empty, then so is the limit $\bigcap_{j\in\bbN}
I_j$. 

With an energy relation~$\calR_\calP^E$, we~associate an
\emph{energy function} (also denoted with~$\calR_\calP^E$, or
simply~$\calR$, as long as no ambiguity may arise), defined for any
closed sub-interval~$I\in \calI(E)$ as
\(
\calR(I)=\{w_1 \in E\mid \exists w_0\in I.\ \calR(w_0,w_1)\}
\).
Symmetrically:
\[
\calR^{-1}(I)=\{w_0\in E \mid \exists w_1\in I.\ \calR(w_0,w_1)\}.
\]
Observe that $\calR(I)$ and $\calR^{-1}(I)$ also belong
to~$\calI(E)$ (because the relation~$\calR$ is closed and convex).
Moreover, $\calR$ and~$\calR^{-1}$ are non-decreasing: 
  for any two intervals $I$ and~$J$ in~$\calI(E)$ such that
  $I\subseteq J$, it~holds $\calR(I)\subseteq \calR(J)$ and
  $\calR^{-1}(I)\subseteq \calR^{-1}(J)$.
%
Energy function~$\calR^{-1}$ also satisfies the following
continuity property:%
\begin{restatable}{lemma}{lemcontinuity}\label{lemma-continuity}
  Let $(I_j)_{j\in\bbN}$ be a chain of intervals of~$\calI(E)$, such
  that $I_j\supseteq I_{j+1}$ for all~$j\in\bbN$. Then
  $\calR^{-1}(\bigcap_{j\in\bbN} I_j) = \bigcap_{j\in\bbN} \calR^{-1}(I_j)$.
\end{restatable}



\subsubsection{Composition and fixpoints of energy functions.}
\label{page-subsec:composition_fixpoints}



Consider a finite sequence of paths~$(\calP_i)_{1\leq i\leq k}$.
Clearly, the energy relation for this sequence can be obtained as the
composition of the individual energy relations $\calR_{\calP_k}^E
\circ \cdots \circ \calR_{\calP_1}^E$; the~resulting energy relation
still is a closed convex subset of~$E\times E$ that can be described
as the conjunction of finitely many linear constraints over~$w_0$
and~$w_1$.  As~a special case, we~write $(\calR_\calP^E)^k$ for the
composition of $k$ copies of the same relations~$\calR_\calP^E$.

Now, using Lemma~\ref{lemma-continuity}, we~easily prove that the
greatest fixpoint $\nu \calR^{-1}$ of~$\calR^{-1}$ in the complete
lattice $(\calI(E), \supseteq)$ exists and equals:
\[
\nu \calR^{-1} = \bigcap_{i\in\bbN} (\calR^{-1})^i(E).
\]
Moreover $\nu \calR^{-1}$ is a closed (possibly empty) interval.  Note
that $\nu \calR^{-1}$ is the maximum subset $S_\calR$ of~$E$ such
that, starting with any $w_0\in S_\calR$, it~is possible to
iterate~$\calR$ infinitely many times (that~is, for~any~$w_0\in
S_\calR$, there exists $w_1\in S_\calR$ such that
$\calR(w_0,w_1)$---any~such set $S$ is a post-fixpoint of $\calR^{-1}$, i.e.
$S\subseteq \calR^{-1}(S)$).

If $\calR$ is the energy relation of a cycle~$\calC$ in
the flat SETA, then $\nu \calR^{-1}$ precisely describes the set of initial
energy levels allowing infinite runs through~$\calC$ satisfying the
energy constraint~$E$.
If~$\calR$ is described as
the conjunction~$\phi_\calC$ of a finite set of linear constraints, then
we~can characterize those intervals $[a,b]\subseteq E$ that constitute
a post-fixpoint for~$\calR^{-1}$ by the following first-order formula:
\begin{equation}
a \leq b \wedge a\in E \wedge b\in E \wedge \forall w_0\in[a;b].\
\exists w_1\in[a;b].\ \phi_\calC(w_0,w_1).
\label{eq:post-fixpoint}
\end{equation}

Applying quantifier elimination (to~$w_0$ and~$w_1$), the~above
formula may be transformed into a direct constraint on~$a$ and~$b$,
characterizing all post-fixpoints of~$\calR^{-1}$.  We~get a
characterization of~$\nu \calR^{-1}$ by computing the values of $a$ and~$b$
that satisfy these constraint and maximize~$b-a$.

\begin{example}\label{ex-iel-seta}
  We~again consider the flat SETA of Fig.~\ref{fig-exeta}, and consider the
  energy constraint~$E = [0;5]$. We~first focus on the cycle~$\calC$ on
  the macro-state~$s_2$:
  using the energy relation computed in Example~\ref{ex-enerel},
our~first-order formula for the fixpoint then reads as follows:
\begin{multline*}
  0\leq a\leq b\leq 5 \wedge
  \forall w_0\in[a;b].\ \exists w_1\in[a;b].\\
\bigl((w_1+2 \leq 2w_0\leq w_1+4) \wedge
(w_1-0.5\leq w_0\leq w_1+1)\bigr).
\end{multline*}
Applying quantifier elimination, we~end up with
\(
2\leq a\leq b\leq 4
\).
The~maximal fixpoint then is~$[2;4]$.
Similarly, for the path $\calP$ from~$s_0$ to~$s_2$:
\begin{multline*}
\calR_\calP^E(w_0,w_1) \iff
\exists d_0, d_1.\ 0\leq d_0\leq 1 \wedge 0\leq d_1\leq 1 \wedge d_0+d_1\geq 1 \wedge{}\\
  0\leq w_0\leq 5 \wedge 0\leq w_0+1\leq 5 \wedge
  w_1=w_1+1-d_1 \wedge 0\leq w_1\leq 5
 \end{multline*}
which reduces to $0\leq w_0\leq 4 \wedge w_0\leq w_1\leq w_0+1$.
Finally, the initial energy levels~$w_0$ for which there is an infinite-run
in the whole SETA are characterized by
\(
\exists w_1.\ (0\leq w_0\leq 4 \wedge w_0\leq w_1\leq w_0+1) \wedge (2\leq w_1\leq 4)
\),
which reduces to $1\leq w_0\leq 4$.
\end{example}


\subsubsection{Algorithm for flat segmented energy timed automata.}

\input{eta-algo}

%% file: eta-algo.tex
Following Example~\ref{ex-iel-seta}, we~now prove that we can solve
the energy-constrained infinite-run problem for any flat SETA.  The
next theorem is crucial for our algorithm:

\begin{restatable}{theorem}{thmterm} \label{th:termination}
  Let $\calR$ be the energy relation of an ETP $\calP$ with energy
  constraint~$E$, and let $I\in\calI(E)$.
  Then
  either $I\cap\nu \calR^{-1} \neq \emptyset$ or $\calR^n(I)=\emptyset$
  for some~$n$.
\end{restatable}


It~follows that the energy-constrained infinite-run problem is
decidable for flat SETAs.  The decision procedure traverses
the underlying graph of~$\calA$, forward propagating an initial energy
interval $I_0 \subseteq E$ looking for a simple cycle~$C$ such that
$\nu \calR_C^{-1} \cap I \neq \emptyset$, where $I \subseteq E$ is
the energy interval forward-propagated until reaching the cycle.
%
\begin{algorithm}[tb]
    \algsetup{linenodelimiter=.}
    \begin{algorithmic}[1]  
    \REQUIRE  A \emph{flat SETA} $\calA=\tuple{S,T,P}$; initial state $m_0 \in S$;
    energy interval $I_0$
    \STATE $W \gets \{(m_0, I_0, c)\}$ \COMMENT{initialize the waiting list}
    \WHILE{$W \neq \emptyset$} 
    	\STATE pick $\tuple{m,I,\flag} \in W$ 	\COMMENT{pick an element from the waiting list}
	\STATE $W \gets W \setminus \tuple{m,I,\flag}$ \COMMENT{remove the element from the waiting list}
	\IF[the node $m$ shall be explored without following a cycle]{$\flag = \bar{c}$}
		\FOR{\textbf{each} $(m,m') \in T$ that is not part of a simple cycle of $(S,T)$}
			\STATE $W \gets W \cup \{ \tuple{m', \calR_{P(m,m')}^E(I), c}\}$ \COMMENT{add this new task to the waiting list}
		\ENDFOR
	\ELSE[the node $m$ shall be explored by following a cycle]
		\IF{$m$ belongs to a cycle of $(S,T)$}
			\STATE let $\calC = (m_1,m_2)\cdots (m_k,m_{k+1})$ be the simple cycle s.t.\ $m = m_1 = m_{k+1}$
			\STATE let $\calR_\calC = \calR_{P(m_k,m_{k+1})} \circ \cdots \circ \calR_{P(m_1,m_{2})}$
			\COMMENT{energy relation of the cycle}
        			\IF[check if there is an infinite run via the cycle $C$]{$I \cap \nu \calR_\calC^{-1} \neq \emptyset$}
        				\RETURN $\mathtt{tt}$
        			\ELSE[the cycle can be executed only finitely many times]
				\STATE $W \gets W \cup \{(m, I, \bar{c})\}$ \COMMENT{add a new task to the waiting list}
        				\STATE $i \gets 0$ \COMMENT{initialize the number of cycle executions}
        				\WHILE[while $i$-th energy relation is satisfied]{$\calR_\calC^i(I) \neq \emptyset$}\label{line:18}
        					\FOR{$1 \leq j < k$} 
						\STATE let $\calR_{\calP_j} = \calR_{P(m_j,m_{j+1})} \circ \cdots \circ \calR_{P(m_1,m_{2})}$ \COMMENT{unfold $C$ up to $m_{j+1}$}
        						\STATE $W \gets W \cup \{ \tuple{m_{j+1}, \calR_{\calP_j}(\calR_C^i(I)), \bar{c}}\}$
						\COMMENT{add a task to the waiting list}
        					\ENDFOR
        					\STATE $i \gets i + 1$ \COMMENT{increment the number of cycle executions}
        				\ENDWHILE\label{line:24}
        			\ENDIF
		\ELSE[$m$ doesn't belong to a cycle]
			\STATE $W \gets W \cup \{(m, I, \bar{c})\}$ \COMMENT{add a new task to the waiting list}
		\ENDIF
	\ENDIF
     \ENDWHILE
     \RETURN $\mathtt{ff}$ \COMMENT{no infinite run could be found}
    \end{algorithmic}
    \caption{Existence of energy-constrained infinite runs in flat SETA}
    \label{alg:infiniterun}\label{alg:infiniteruns}
\end{algorithm}
\newcommand{\node}{macro-state}
Algorithm~\ref{alg:infiniterun} gives a detailed description of the
decision procedure.
It~traverses the underlying graph $(S,T)$ of
the flat SETA~$\calA$, using a waiting list~$W$ to
keep track of the \node s that need to be further
explored. The~list~$W$ contains tasks of the form $\tuple{m,I,\flag}$
where $m \in S$ is the current \node, ${I \in\calI(E)}$~is the current
energy interval, and $\flag \in \{c, \bar{c}\}$ is a flag
indicating if $m$~shall be explored by following a cycle it belongs to
($\flag = c$), or proceed by exiting that cycle ($\flag
= \bar{c}$).
Theorem~\ref{th:termination} ensures termination of the \textbf{while}
loop of lines~\ref{line:18}-\ref{line:24}, whereas flatness ensures
the correctness of Algorithm~\ref{alg:infiniterun}.

It is worth noting that the flatness assumption for the
SETA~$\calA$ implies that the graph~$(S,T)$ has finitely many cycles
(each~macro-state belongs to at most one simple cycle of~$(S,T)$,
therefore the number of cycles is bounded by the number of
macro-states). As a consequence, Algorithm~\ref{alg:infiniterun}
performs in the worst case an exhaustive search of all cycles
in~$\calA$. The~technique does not trivially extend to SETAs with
nested cycles, because they may have infinitely many cycles.

%% file: etau.tex
The   assumptions   of   perfect   knowledge   of   energy-rates   and
energy-updates  are often  unrealistic, as  is the  case in  the HYDAC
oil-pump  control problem  (see~Section~\ref{sec-hydac}). Rather,  the
knowledge  of energy-rates  and  energy-updates comes  with a  certain
imprecision,  and the  existence of  energy-constrained infinite  runs
must take these into account in  order to be robust.  In~this section,
we~revisit the  energy-constrained infinite-run problem in  the setting
of imprecisions, by~viewing it as a two-player game problem.

\subsubsection{Adding uncertainty to ETA.}
%
  An \emph{energy timed automaton with \textbf{uncertainty}}~(ETAu
  for~short) is a
  tuple~$\calA=\tuple{S,S_0,\Cl,\inv,\rate,T,\epsilon,\Delta}$, where
  $\tuple{S,S_0,\Cl,\inv,\rate,T}$ is an energy timed automaton, with
  $\epsilon\colon S\to \bbQ_{>0}$ assigning imprecisions to rates of
  states and $\Delta\colon T\to\bbQ_{>0}$ assigning imprecisions to
  updates of transitions.
%
This notion of uncertainty extends to
\emph{energy timed path with uncertainty}~(ETPu) and to
\emph{segmented energy timed automaton with uncertainty}~(SETAu).

\smallskip
Let $\calA=\tuple{S,S_0,\Cl,\inv,\rate,T,\epsilon,\Delta}$ be an ETAu,
and let $\tau = (t_i)_{0 \le i < n}$ be a finite sequence of
transitions, with $t_i = (s_i,g_i,u_i,z_i,s_{i+1})$ for every $i$.
A~finite \emph{run} in $\calA$ on $\tau$ is a sequence of
configurations $\rho=(\ell_j,v_j,w_j)_{0\leq j\leq 2n}$ such that
there exist a sequence of delays~$d=(d_i)_{0\leq i<n}$ for which the
following requirements hold:
\begin{itemize}
\item for all $0\leq j<n$, $\ell_{2j}=\ell_{2j+1}=s_j$, and
  $\ell_{2n}=s_n$;
\item for all $0\leq j<n$, $v_{2j+1}=v_{2j}+d_j$ and
  $v_{2j+2}=v_{2j+1}[z_j\to 0]$;
\item for all $0\leq j<n$, $v_{2j}\models \inv(s_j)$ and
  $v_{2j+1}\models \inv(s_j) \wedge g_j$;
\item for all $0\leq j<n$, it holds that $w_{2j+1}=w_{2j} + d_j\cdot
  \alpha_j$ and $w_{2j+2}=w_{2j+1} + \beta_j$, where
  $\alpha_j\in[r(s_j)-\epsilon(s_j), r(s_j)+\epsilon(s_j)]$ and
  $\beta_j\in[u_j-\Delta(t_j), u_j+\Delta(t_j)]$.
\end{itemize}
%
We say that $\rho$ is a possible outcome of $d$ along $\tau$, and that
$w_{2n}$ is a possible final energy level for $d$ along $\tau$, given
initial energy level $w_0$.  Note that due to uncertainty, a~given
delay sequence $d$ may have several possible outcomes (and
corresponding energy levels) along a given transition sequence $\tau$.
In~particular, we~say
that $\tau$ together with $d$ and initial energy level~$w_0$ satisfy
an energy constraint $E\in \calI(\bbQ)$ if any possible outcome run
$\rho$ for $t$ and $d$ starting with~$w_0$ satisfies~$E$. All these
notions are formally extended to ETPu.

Given an ETPu $\calP$, and a delay sequence $d$ for $\calP$ satisfying
a given energy constraint $E$ from initial level $w_0$, we denote by
$\calE^E_{\calP,d}(w_0)$ the set of possible final energy levels.
It~may be seen that $\calE^E_{\calP,d}(w_0)$ is a closed subset
of~$E$.

\begin{example}
  Figure~\ref{fig-unctp} is the energy timed path~$\calP$ of
  Fig.~\ref{fig-extp} extended with uncertainties of~$\pm 0.1$ on all
  rates and updates.  The~runs associated with path~$\calP$, delay
  sequence~$d=\tuple{0.6,0.4}$ and initial energy level~$w_0=3$
  satisfy the energy constraint~$E=[0;5]$. The~set
  $\calE^E_{\calP,d}(w_0)$ then is $[2.5; 3.1]$.
%
%
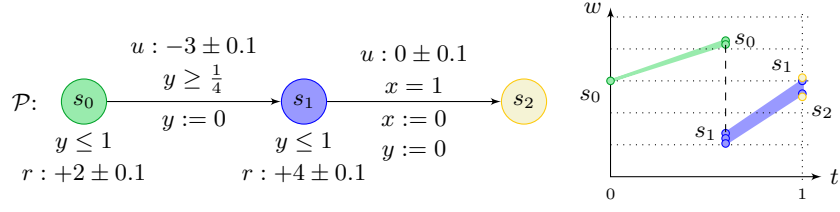
\begin{figure}[t]
  \centering
  \begin{tikzpicture}
  \begin{scope}[scale=.65]
  \draw (0,0) node[rond6,vert] (a) {$s_0$}
      node[below=3mm] {$\stack{y\leq 1}{\rate:+2\pm0.1}$} node[left=5mm] {$\calP$:};
    \draw (4.5,0) node[rond6,bleu] (b) {$s_1$}
      node[below=3mm] {$\stack{y\leq 1}{\rate:+4\pm0.1}$};
    \draw (9,0) node[rond6,jaune,opacity=1] (c) {$s_2$};
    \draw (a) edge[-latex'] node[pos=.5,above] {$y\geq \frac14$}
      node[pos=.5,above=5mm] {$\update:-3\pm0.1$}
      node[below] {$y:=0$} (b);
    \draw (b) edge[-latex'] node[pos=.5,above] {$x=1$} node[below] {$x:=0$} node[below=4mm] {$y:=0$}
      node[above=4mm, pos=.5] {$\update: 0\pm0.1$}  (c);
      \end{scope}

    \begin{scope}[yshift=-1cm,xshift=7cm,xscale=3,yscale=.5,scale=.85]
      \fill[fvert,line width=.4mm,draw=fvert] (0,3) -- (.6,4.26) -- (.6,4.14) -- cycle;
      \fill[fbleu, line width=.4mm,draw=fbleu] (0.6,1.36) -- (1,3) -- (1,2.6) -- (0.6,1.04) -- (0.6,1.36);
      \draw[latex'-latex'] (0,5.3) node[left] {$w$} |- (1.1,0) node[right] {$t$};
      \foreach \i in {1,2,3,4,5}
        {\draw[dotted] (0,\i) -- +(1.05,0);}
      \draw[dotted] (1,-.1) node[below] {$\scriptstyle 1$} -- +(0,5.3);
      \path (0,-.1)  node[below] {$\scriptstyle 0$};
      \draw (0,3) node[rond1,vert] (a) {} node[below left] {$s_0$};  
      \draw (0.6,4.26) node[rond1,vert] (b) {} node[right] {$s_0$};
      \draw (0.6,4.14) node[rond1,vert] (b') {};
      \draw (0.6,1.36) node[rond1,bleu] (c) {} node[left] {$s_1$};
      \draw (0.6,1.2) node[rond1,bleu] {};
      \draw (0.6,1.04) node[rond1,bleu] (c') {};
      \draw (1,3) node[rond1,bleu] (d) {} node[above left] {$s_1$};
      \draw (1,2.6) node[rond1,bleu] (d') {};
      \draw (1,3.1) node[rond1,jaune] (e) {};
      \draw (1,2.5) node[rond1,jaune] (e') {}  node[below right] {$s_2$};
      \draw[dashed] (b) -- (c');
      \end{scope}

\end{tikzpicture}
  \caption{An energy timed path~$\calP$ with
    uncertainty, and a representation of the runs corresponding to the
    delay sequence $\protect\tuple{0.6, 0.4}$ with initial energy level~$3$.}
  \label{fig-unctp}
\end{figure}
\end{example}

Now let $\calA=\tuple{S,T,P}$ be an SETAu and let $E$ be an energy
constraint.  A~(memoryless\footnote{For the infinite-run problem,
it can be shown that memoryless strategies suffice.})
\emph{strategy} $\sigma$ returns for any macro-configuration $(s,w)$
($s\in S$ and $w\in E$) a pair $(t,d)$, where $t=(s,s')$ is a
successor edge in $T$ and $d\in\bbR+^n$ is a delay sequence for the
corresponding energy timed path, i.e. $n=|P(t)|$.  A (finite or
infinite) execution of $(\rho^i)_i$ writing $\rho^i=(\ell^i_j, x^i_j,
w^i_j)_{0\leq j\leq 2n_i}$, is an outcome of $\sigma$ if the following
conditions hold:
\begin{itemize}
\item $s^i_0$ and $s^i_{2n_i}$ are macro-states of~$\calA$, and
  $\rho^i$ is a possible outcome of $P(s^i_0,s^i_{2n_i})$ for $d$
  where $\sigma(s^i_0,w^i_0)=\big((s^i_0,s^i_{2n_i}),d\big)$;
\item $s^{i+1}_0 = s^i_{2n_i}$ and $w^{i+1}_0=w^i_{2n_i}$.
\end{itemize}
%
Now  we may  formulate  the  infinite-run  problem in  the
setting of uncertainty:
  for a SETAu~$\calA$, an energy constraint~$E\in\calI(\bbQ)$,
  and a macro-state~$s_0$ and an initial energy level~$w_0$
  the \emph{energy-constrained
    infinite-run problem} is to decide the existence of a strategy
  $\sigma$ for $\calA$ such that all runs $(\rho^i)_i$ that are
  outcome of $\sigma$ starting from configuration $(s_0,w_0)$ satisfy
  $E$?


\subsubsection{Ternary energy relations.}





Let $\calP=(\{s_i \mid 0 \le i \le
n\},\{s_0\},X,I,r,T,\epsilon,\Delta)$ be an ETPu and let
$E\in\calI(\bbQ)$ be an energy constraint.  The ternary energy
relation $\calU^E_\calP\subseteq E\times E\times E$ relates all
triples $(w_0,a,b)$ for which there is a strategy $\sigma$ such that
any outcome of $\rho$ from $(s_0,\mathbf{0},w_0)$ satisfies $E$ and
ends in a configuration $(s_n,\mathbf{0},w_1)$ where $w_1 \in[a;b]$.
This relation can be characterized by the following first-order
formula:
\begin{multline*}
  \calU_\calP^E(w_0,a,b) \iff 
     \exists (d_i)_{0\leq i<n}.\
     \Phi_{\text{timing}} \wedge \Phi_{\text{energy}}^i \wedge {}\\
    \qquad w_0+\sum_{k=0}^{n-1}(r(s_k)\cdot d_k + u_k) +
            \sum_{k=0}^{n-1}([-\epsilon(s_k);\epsilon(s_k)]\cdot d_k
            + [-\Delta(t_k); \Delta(t_k)]) \subseteq [a;b]
\end{multline*}
where $ \Phi_{\text{energy}}^i$ encodes the energy constraints as the
inclusion of the interval of reachable energy levels in the energy
constraint (in~the same way as we do on the second line of the
formula). Interval inclusion can then be expressed as constraints on
the bounds of the intervals.
%
It~is clear  that $\calU^E_{\calP}$ is a closed, convex
subset  of $E\times  E\times  E$  and can  be  described  as a  finite
conjunction  of  linear  constraints  over  $w_0,  a$  and  $b$  using
quantifier elimination.

\begin{example}\label{ex7}
  We   illustrate    the   above   translation   on    the   ETPu   of
  Fig.~\ref{fig-unctp}.   For energy  constraint  $[0;5]$, the  energy
  relation can be written as:
\begin{xalignat*}1
& \calU_\calP^E(w_0,a,b) \iff \exists d_0,d_1.\  d_0\in[0.25;1] \wedge d_1\in [0;1] \wedge d_0+d_1=1 \wedge w_0\in[0;5] \wedge{}\\
& \qquad w_0+ [1.9; 2.1]\cdot d_0 \subseteq [0;5]  \wedge{}\\
& \qquad w_0+ [1.9; 2.1]\cdot d_0 +[-3.1;-2.9]\subseteq [0;5]  \wedge{}\\
& \qquad w_0+ [1.9; 2.1]\cdot d_0 +[-3.1;-2.9] + [3.9;4.1]\cdot d_1\subseteq [0;5]  \wedge{}\\
& \qquad w_0+ [1.9; 2.1]\cdot d_0 +[-3.1;-2.9] + [3.9;4.1]\cdot d_1+[-0.1;0.1]\subseteq [a;b]\subseteq [0;5]] 
\end{xalignat*}
%
Applying quantifier elimination, we end up with:
\begin{xalignat*}1
\calU_\calP^E(w_0,a,b) \iff &
0\leq a\leq b\leq 5 \wedge b\geq a+0.6 \wedge
a-0.2 \leq w_0 \leq b+0.7 \wedge{} \\
& (4.87+1.9\cdot a)/3.9  \leq w_0 \leq (7.27+2.1\cdot b)/4.1
\end{xalignat*}
We~can use this relation in order to compute the set of initial energy
levels from which there is a strategy to end up in~$[2.5;3.1]$ (which
was the set of possible final energy levels in the example of
Fig.~\ref{fig-unctp}). We~get $w_0\in[37/15;689/205]$, which is
(under-)approximately $w_0\in[2.467;3.360]$.
\end{example}

\subsubsection{Algorithm for SETAu.}

Let $\calA=(S,T,P)$ be a SETAu and let $E\in\calI(\bbQ)$ be an energy
constraint. Let $\calW\subseteq S\times E$ be the maximal set of
configurations satisfying the following:
\begin{align}
  (s,w)\in\calW \,\Rightarrow & \,\exists t=(s,s')\in T. \exists a,b \in
                       E. \nonumber\\
  & \,\,\,\calU^E_{P(t)}(w,a,b) 
    \wedge \forall w'\in[a;b]. (s',w')\in\calW
\label{Weq}
\end{align}
Now $\calW$ is easily shown to characterize the set of configurations
$(s,w)$ that satisfy the energy-constrained infinite-run problem.
%
%
Unfortunately this characterization does not readily provide an
algorithm.  We~thus make the following restriction and show that
it leads to decidability of the energy-constrained infinite-run
problem:
\begin{description}
\item[(R)]\label{restr} in any of the ETPu $P(t)$ of $\calA$, on at least one of
  its transitions, some clock $x$ is compared with a positive lower
  bound. Thus, there is an (overall minimal) positive time-duration
  $D$ to complete any $P(t)$ of $\calA$.
\end{description}

\begin{restatable}{theorem}{thmuncert}
  \label{unc-thm}
 The energy-constrained infinite-run problem is decidable for SETAu
    satisfying~{\bf{(R)}}.
\end{restatable}

It is worth noticing that we do \textbf{not} assume flatness of the
model for proving the above theorem. Instead, the minimal-delay
assumption~{\bf{(R)}} has to be made: it~entails that any stable set
is made of intervals whose size is bounded below, which provides an
upper bound on the number of such intervals. We~can then rewrite the
right-hand-size expression of~\eqref{Weq} as:
\begin{align}
    \bigwedge_{s\in S}\bigwedge_{1\leq j\leq N} & [a_{s,j};b_{s,j}]
    \subseteq E \wedge w_0 \in \bigvee_{1 \le j \le N}
    [a_{s_0,j};b_{s_0,j}] \wedge \forall
    w\in[a_{s,j};b_{s,j}]. \nonumber\\
    & \bigvee_{(s,s')\in T}\big[ \exists a,b\in E.
    \,\,\calU^E_{P(s,s')}(w,a,b)\wedge \bigvee_{1\leq k\leq
      N}([a;b]\subseteq [a_{s',k};b_{s',k}])\big]
      \label{Ueq}
    \end{align}

\begin{example}
We pursue on Example~\ref{ex7}.
If ETPu $\calP$ is iterated (as on the loop on state~$s_2$ of
Fig.~\ref{fig-exeta}, but now with uncertainty), the set $\calW$
(there is a single macro-state) can be captured with a single interval
$[a,b]$. We characterize the set of energy levels from which the path
$\calP$ can be iterated infinitely often while satisfying the energy
constraint $E = [0;5]$ using equation~\eqref{Ueq}, as follows:
\[
0\leq a\leq b\leq 5 \wedge
\forall w_0\in[a;b].~\calU_\calP^E(w_0,a,b).
\]
We~end up with
\(
2.435\leq a \wedge b\leq 3.635 \wedge b\geq a+0.6
\),
so that the largest interval is $[2.435;3.635]$ (which can be compared
to the maximal fixpoint $[2;4]$ that we obtained in
Example~\ref{ex-iel-seta} for the same cycle without uncertainty).
\end{example}

As in the setting without uncertainties, we~can also synthesize an
(optimal) upper-bound for the energy constraint:



\begin{restatable}{theorem}{thmubuncert}
  Let $\calA=(S,T,P)$ be a depth-1 flat SETAu satisfying {\bf{(R)}}.
  Let $L\in\bbQ$ be an energy lower bound, and let $(s_0,w_0)$ be an
  initial macro-configuration. Then the existence of an upper energy
  bound $U$, such that the energy-constrained infinite-run problem is
  satisfied for the energy constraint $[L;U]$ is decidable.
  Furthermore, one can compute the least upper bound, if there is one.
\end{restatable}

%% file: hydac-short.tex
\subsubsection{Modelling the Oil Pump System.}
\newcommand{\on}{\texttt{On}}
\newcommand{\off}{\texttt{Off}}

In this section we describe the characteristics of each component of
the HYDAC case, which we then model as a~SETA.
\begin{trivlist}
\item \emph{The Machine.} The oil consumption of the machine is
  cyclic. One cycle of consumptions, as given by HYDAC, consists of
  $10$ periods of consumption, each having a duration of two seconds,
  as~depicted in Figure~\ref{fig:machinecycle}. Each period is
  described by a rate of consumption~$m_r$ (expressed in litres per
  second).
  The~consumption rate is subject to noise: if the mean consumption
  for a period is $c\;\si{\litre/\second}$ (with~$c \geq 0$) its
  actual value lies within $[\max(0,c -\epsilon); c + \epsilon]$,
  where $\epsilon$ is fixed to $0.1\;\si{\litre/\second}$.

\item \emph{The Pump.} The pump is either \on\ or \off, and we assume
  it is initially \off\ at the beginning of a cycle.  While it is~\on,
  it~pumps oil into the accumulator with a rate~$p_r =
  2.2\;\si{\litre/\second}$.  The~pump is also subject
  to timing constraints, which prevent switching it on and off too often. 
%
\item \emph{The Accumulator.} The volume of oil within the accumulator
  will be modelled by means of an energy variable~$v$.
%
  Its~evolution is given by the differential inclusion $dv/dt -u\cdot
  p_r\in -[m_r+\epsilon; m_r-\epsilon]$ (or~$-[m_r+\epsilon;0]$ if
  $m_r-\epsilon<0$), where $u\in\{0,1\}$ is the state of the pump.
\end{trivlist}
The controller must operate the pump (switch it on and~off) to ensure
the following requirements:
\begin{enumerate*}[label=(R\arabic*)]
\item the level of oil in the accumulator must always stay within the
  safety bounds $E = [V_{\min}; V_{\max}]=[4.9;25.1]\litr$
\item the average level of oil in the accumulator is kept as low as possible.
\end{enumerate*}

By modelling the oil pump system as a SETA~$\mathcal{H}$, the above
control problem can be reduced to finding a deterministic schedule
that results in a safe infinite run in~$\mathcal{H}$. Furthermore, we
are also interested in determining the minimal safety interval~$E$,
i.e., finding interval bounds that minimize~$V_{\max} - V_{\min}$,
while ensuring the existence of a valid controller for~$\mathcal{H}$.

As a first step in the definition of~$\calH$, we~build an ETP
representing the behaviour of the machine, depicted on Fig.~\ref{fig-setaconsum}.
\begin{figure}[ht]
  \centering
  \begin{tikzpicture}[xscale=1.3]
    \begin{scope}
\everymath{\scriptstyle}
\path[use as bounding box] (-.4,.3) -- (9.4,-.5);
\draw (0,0) node[rond5,orange] (a) {} node {$0$} node[below=3mm] {$x\leq 2$};
\draw (1,0) node[rond5,orange] (b) {} node {$-1.2$} node[below=3mm] {$x\leq 2$};
\draw (2,0) node[rond5,orange] (c) {} node {$0$} node[below=3mm] {$x\leq 2$};
\draw (3,0) node[rond5,orange] (d) {} node {$0$} node[below=3mm] {$x\leq 2$};
\draw (4,0) node[rond5,orange] (e) {} node {$-1.2$} node[below=3mm] {$x\leq 2$};
\draw (5,0) node[rond5,orange] (f) {} node {$-2.5$} node[below=3mm] {$x\leq 2$};
\draw (6,0) node[rond5,orange] (g) {} node {$0$} node[below=3mm] {$x\leq 2$};
\draw (7,0) node[rond5,orange] (h) {} node {$-1.7$} node[below=3mm] {$x\leq 2$};
\draw (8,0) node[rond5,orange] (i) {} node {$-0.5$} node[below=3mm] {$x\leq 2$};
\draw (9,0) node[rond5,orange] (j) {} node {$0$} node[below=3mm] {$x\leq 2$};
\draw[-latex'] (a) -- (b) node[midway,above] {$x=2$} node[midway,below] {$x:=0$};
\draw[-latex'] (b) -- (c) node[midway,above] {$x=2$} node[midway,below] {$x:=0$};
\draw[-latex'] (c) -- (d) node[midway,above] {$x=2$} node[midway,below] {$x:=0$};
\draw[-latex'] (d) -- (e) node[midway,above] {$x=2$} node[midway,below] {$x:=0$};
\draw[-latex'] (e) -- (f) node[midway,above] {$x=2$} node[midway,below] {$x:=0$};
\draw[-latex'] (f) -- (g) node[midway,above] {$x=2$} node[midway,below] {$x:=0$};
\draw[-latex'] (g) -- (h) node[midway,above] {$x=2$} node[midway,below] {$x:=0$};
\draw[-latex'] (h) -- (i) node[midway,above] {$x=2$} node[midway,below] {$x:=0$};
\draw[-latex'] (i) -- (j) node[midway,above] {$x=2$} node[midway,below] {$x:=0$};
    \end{scope}
  \end{tikzpicture}
  \caption{The ETP representing the oil consumption of the machine.}
  \label{fig-setaconsum}
  \medskip

  \begin{tikzpicture}
    \everymath{\scriptstyle}
\path[use as bounding box] (-.4,.3) -- (6.4,-.5);
\draw (0,0) node[rond5,orange,opacity=.9,dashed] (a) {} node {$-m$} node[below=3mm] {$x\leq 2$};
    \draw (2,0) node[rounded corners=2mm,minimum width=1cm,minimum height=5mm,rouge] (b) {$p-m$} node[below=3mm] {$x\leq 2$};
    \draw (4,0) node[rounded corners=2mm,minimum width=1cm,minimum height=5mm,vert] (c) {$-m$} node[below=3mm] {$x\leq 2$};
    \draw (6,0) node[rond5,jaune,opacity=.9,dashed] (d) {$-m'$} node[below=3mm] {$x\leq 2$};
    \draw[-latex'] (a) -- (b);
    \draw[-latex'] (b) -- (c);
    \draw[-latex'] (c) -- (d) node[midway,above] {$x=2$} node[midway,below] {$x:=0$};
  \end{tikzpicture}
  \caption{An ETP for modelling the pump}
  \label{fig-setapump}
\end{figure}
In~order to fully model the behaviour of our oil-pump system, one
would require the parallel composition of this ETP with another
ETP representing the pump. The~resulting ETA would not be a flat SETA,
and is too large to be handled by our algorithm with uncertainty.
Since it~still provides interesting
results, we~develop this (incomplete) approach in 
Appendix~\ref{app:hydacfull}.

Instead, we~consider an alternative model of the pump, which only allows
to switch it on and off once during each 2-second slot. This~is
modelled by inserting, between any two states of the model of
Fig.~\ref{fig-setaconsum}, a~copy of the ETP depicted on
Fig.~\ref{fig-setapump}. In~that ETP, the state with rate~$p-m$
models the situation when the pump is~on. Keeping the pump off for the
whole slot can be achieved by spending delay~zero in that state.
%
%
We~name~$\calH_1=\tuple{M,T,P_1}$ the SETA made of a single
macro-state equipped with a self-loop labelled with the ETP above.

In~order to represent the timing constraints of the pump
switches, we~also consider a second SETA
model~$\calH_2=\tuple{M,T,P_2}$ where the pump can be operated only
during every other time slot. This amounts to inserting the ETP of
Fig.~\ref{fig-setapump} only after the first, third, fifth, seventh
and ninth states of the ETP of Fig.~\ref{fig-setaconsum}.


We~also consider extensions of both models with uncertainty
$\epsilon=0.1\;\si{\litre/\second}$ (changing any negative rate~$-m$
into rate interval $[-m-\epsilon;-m+\epsilon]$, but changing rate~$0$
into $[-\epsilon;0]$). We~write~$\calH_1(\epsilon)$
and~$\calH_2(\epsilon)$ for the corresponding models.


\subsubsection{Synthesizing controllers.}

For each model, we synthesize minimal upper bounds~$U$ (within the
interval~$[V_{\min};V_{\max}]$) that admit a solution to the
energy-constrained infinite-run problem for the energy constraint $E =
[V_{\min};U]$. Then, we~compute the greatest stable interval~$[a;b] \subseteq
[L;U]$ of the cycle witnessing the existence of an $E$-constrained
infinite-run.  This~is done by following the methods described
in Sections~\ref{sec:ETA} and~\ref{sec:ETAu} where quantifier 
elimination is performed using Mjollnir~\cite{Monniaux10}.

Finally for each model we synthesise \emph{optimal} strategies that,
given an initial volume $w_0 \in [a,b]$ of the accumulator, return a
sequence of pump activation times~$t_i^{\text{on}}$ and~$t_i^{\text{off}}$ to be 
performed during the cycle.  
This~is performed in two steps: first we encode the set of safe
\emph{permissive strategies} as a quantifier-free first-order linear formula 
having as free variables $w_0$, and the times $t_i^{\text{on}}$ and $t_i^{\text{off}}$.
The formula is obtained by relating $w_0$, and the times $t_i^{\text{on}}$ and $t_i^{\text{off}}$
with the intervals~$[L;U]$ and $[a;b]$ and delays~$d_i$ as prescribed by the
energy relations presented in Sections~\ref{sec:ETA} and~\ref{sec:ETAu}. We use 
Mjollnir~\cite{Monniaux10} to eliminate the existential quantifiers on the delays~$d_i$.
Then, given an energy value~$w_0$ we determine an optimal safe strategy for~it
(i.e.,~some timing values when the pump is turned on and~off) as the solution of the 
optimization problem that minimizes the average oil volume in the tank during one consumption cycle
subject to the \emph{permissive strategies} constraints. To~this~end, we~use the function 
\texttt{FindMinimum} of Mathematica~\cite{Mathematica} to minimize the non-linear cost function 
expressing the average oil volume subject to the linear constraints obtained above.
Fig.~\ref{fig:strategies} shows the resulting strategies:
there, each horizontal line at a given initial oil level indicates the
delays (green intervals) where the pump will be running.



\begin{figure}[t]
\centering
\begin{subfigure}[b]{0.48\textwidth}
 \includegraphics[width=1 \textwidth]{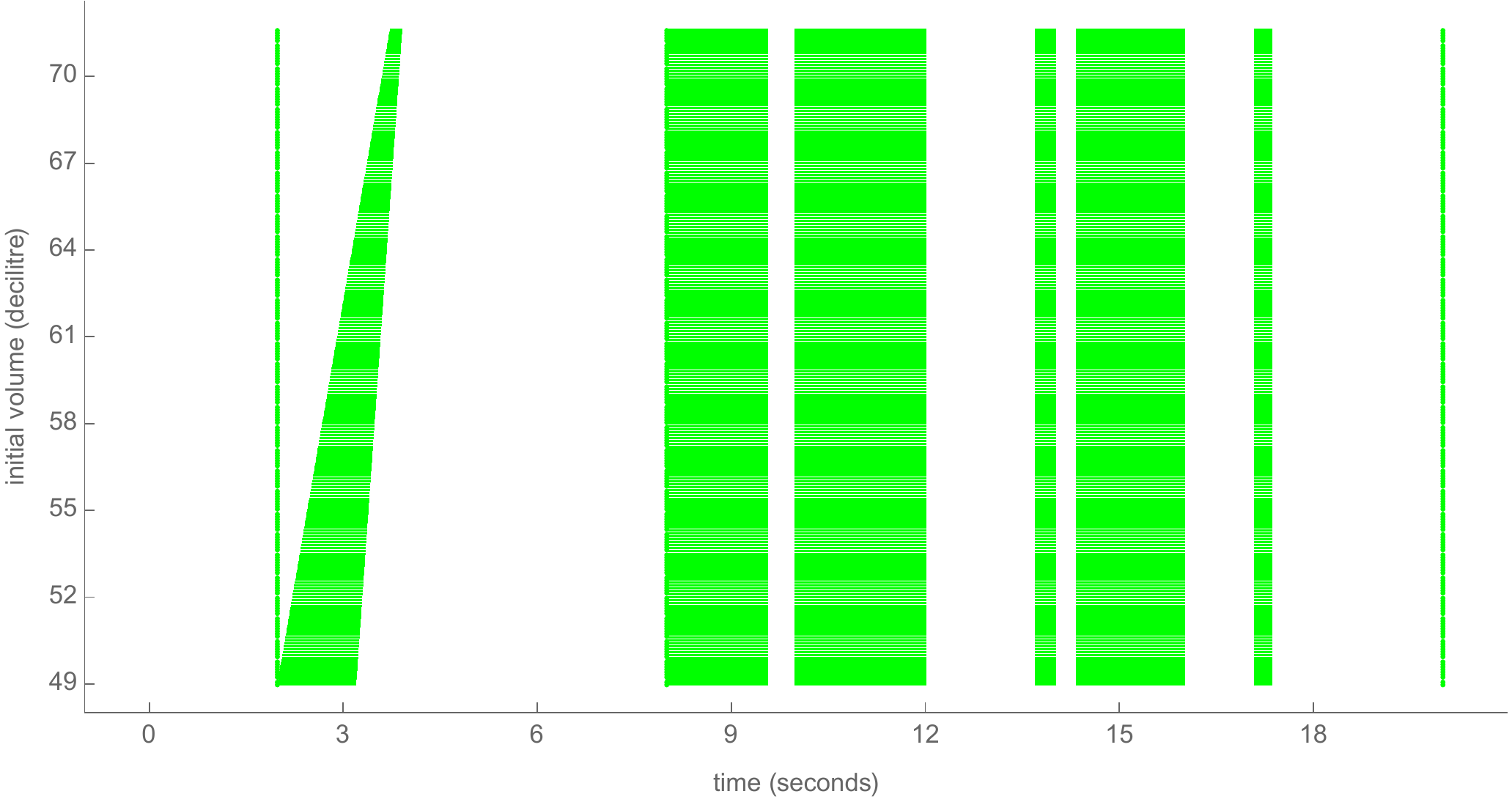}
\end{subfigure}
\quad
\begin{subfigure}[b]{0.48\textwidth}
 \includegraphics[width=1 \textwidth]{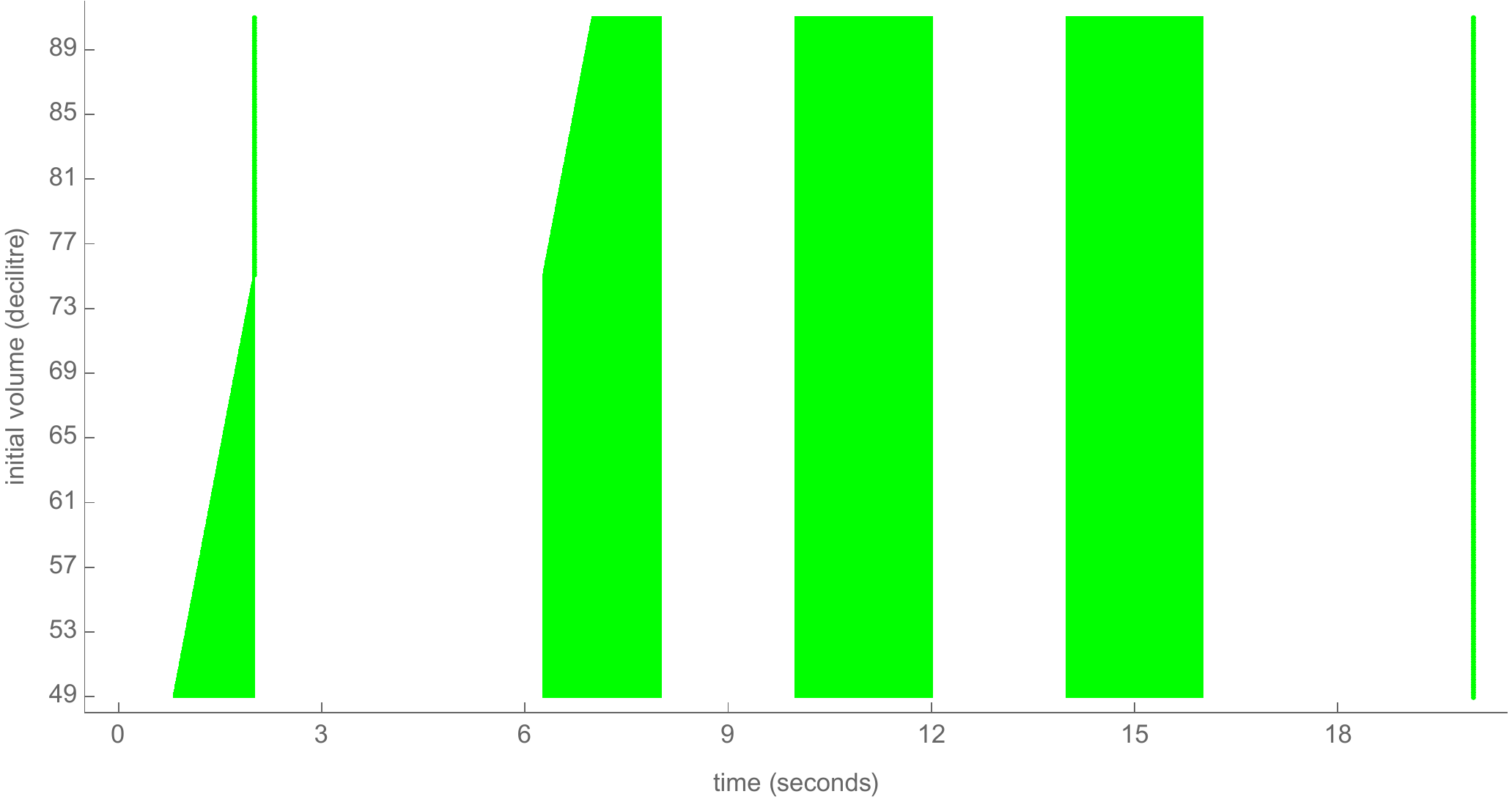}
\end{subfigure}
 \caption{Local strategies for~$\calH_1(\epsilon)$ (left) and $\calH_2(\epsilon)$ (right) for a single cycle of the HYDAC system.}
\label{fig:strategies}
\end{figure}


\begin{table}[t]
\bgroup\centering
\setlength{\tabcolsep}{1ex}
\def\arraystretch{1.2}
\begin{tabular}[t]{|c|c|c|c|}
\hline
Controller & $[L;U]$ & $[a;b]$ & Mean vol. (\si{\litre}) \\ \hline\hline
$\mathcal{H}_1$ & $[4.9;5.84]$ & $[4.9;5.84]$ & 5.43 \\ \hline
$\mathcal{H}_1(\epsilon)$ & $[4.9; 7.16]$ & $[5.1; 7.16]$ & 6.15 \\ \hline
$\mathcal{H}_2$ & $[4.9;7.9]$ & $[4.9;7.9]$ & 6.12 \\ \hline
$\mathcal{H}_2(\epsilon)$ & $[4.9;9.1]$ & $[5.1;9.1]$ & 7.24 \\ \hline 
\hline
G1M1~\cite{CJLRR09} & $[4.9;25.1]^{(*)}$ & $[5.1;9.4]$ & 8.2 \\ \hline
G2M1~\cite{CJLRR09} & $[4.9;25.1]^{(*)}$ & $[5.1;8.3]$ & 7.95 \\ \hline
\hline
\cite{DBLP:conf/fm/ZhaoZKL12} & $[4.9;25.1]^{(*)}$ & $[5.2;8.1]$ & 7.35 \\ \hline
\end{tabular}
\par\egroup
$^{(*)}$ \footnotesize Safety interval given by the HYDAC company.
\medskip
\caption{Characteristics of the synthesized strategies, compared with the 
strategies proposed in~\cite{CJLRR09,DBLP:conf/fm/ZhaoZKL12}.}
\label{tab:schedulers}
\end{table}
Table~\ref{tab:schedulers} summarizes the results
obtained for our models. It~gives the optimal volume constraints, the
greatest stable intervals, and the values of the worst-case (over all initial oil levels in~$[a;b]$)
mean volume.
It~is worth noting that the models without
uncertainty outperform the respective version with uncertainty.
Moreover, the~worst-case mean volume obtained both for
$\mathcal{H}_1(\epsilon)$ and $\mathcal{H}_2(\epsilon)$ are significantly
better than the optimal strategies synthesized both in~\cite{CJLRR09}
and~\cite{DBLP:conf/fm/ZhaoZKL12}.

The reason for this may be that
\begin{enumerate*}[label=(\roman*)]
\item our models relax the latency requirement for the pump, 
\item the strategies of~\cite{CJLRR09} are obtained using a
  discretization of the dynamics within the system, and
\item the strategies of~\cite{CJLRR09}
  and~\cite{DBLP:conf/fm/ZhaoZKL12} were allowed to activate the pump
  respectively two and three times during each cycle.
\end{enumerate*}


\begin{figure}[t]
\centering
\begin{subfigure}[b]{0.48\textwidth}
  \includegraphics[width=\textwidth,height=4.2cm]{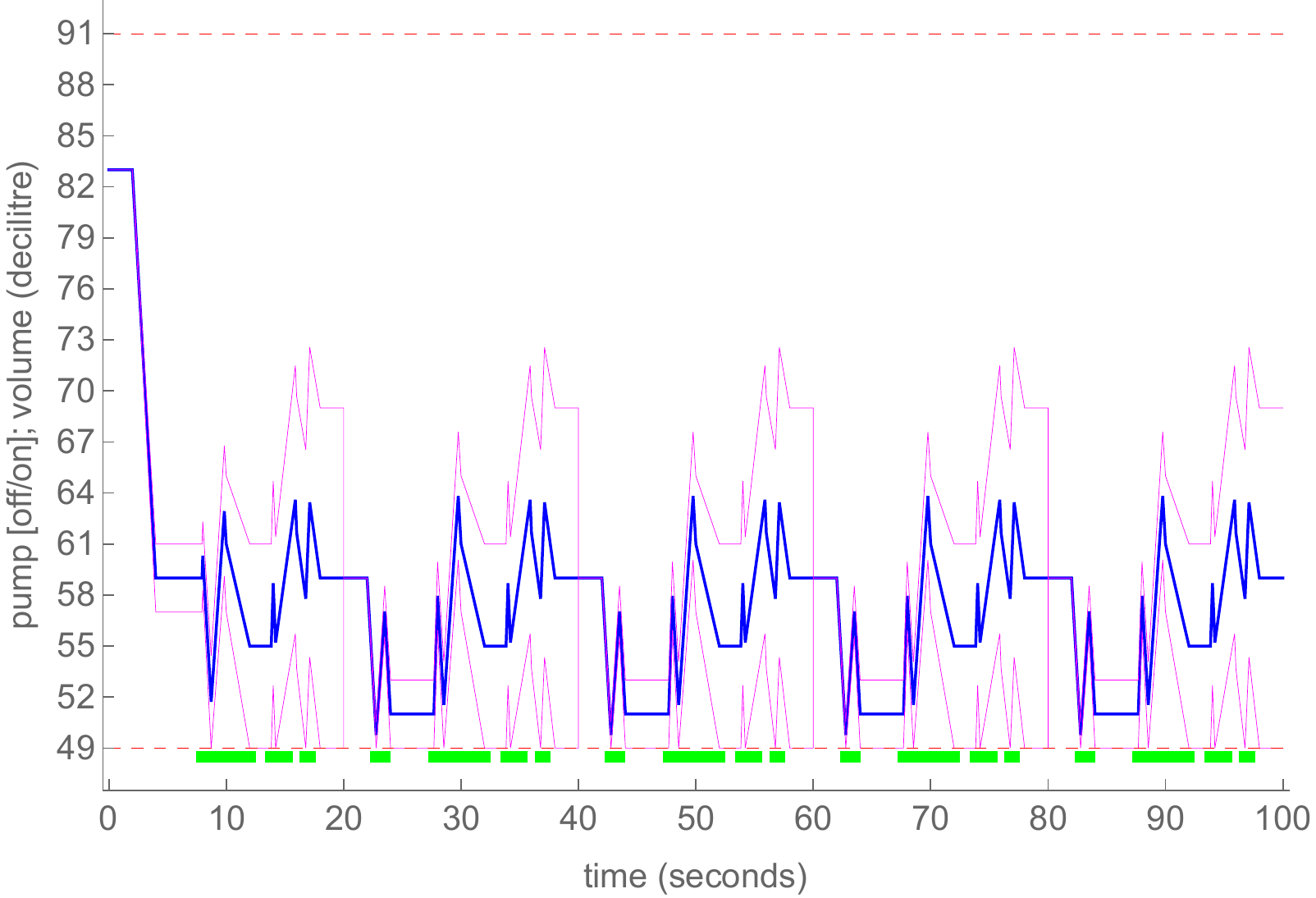}
\end{subfigure}
\quad
\begin{subfigure}[b]{0.48\textwidth}
 \includegraphics[width=\textwidth,height=4.2cm]{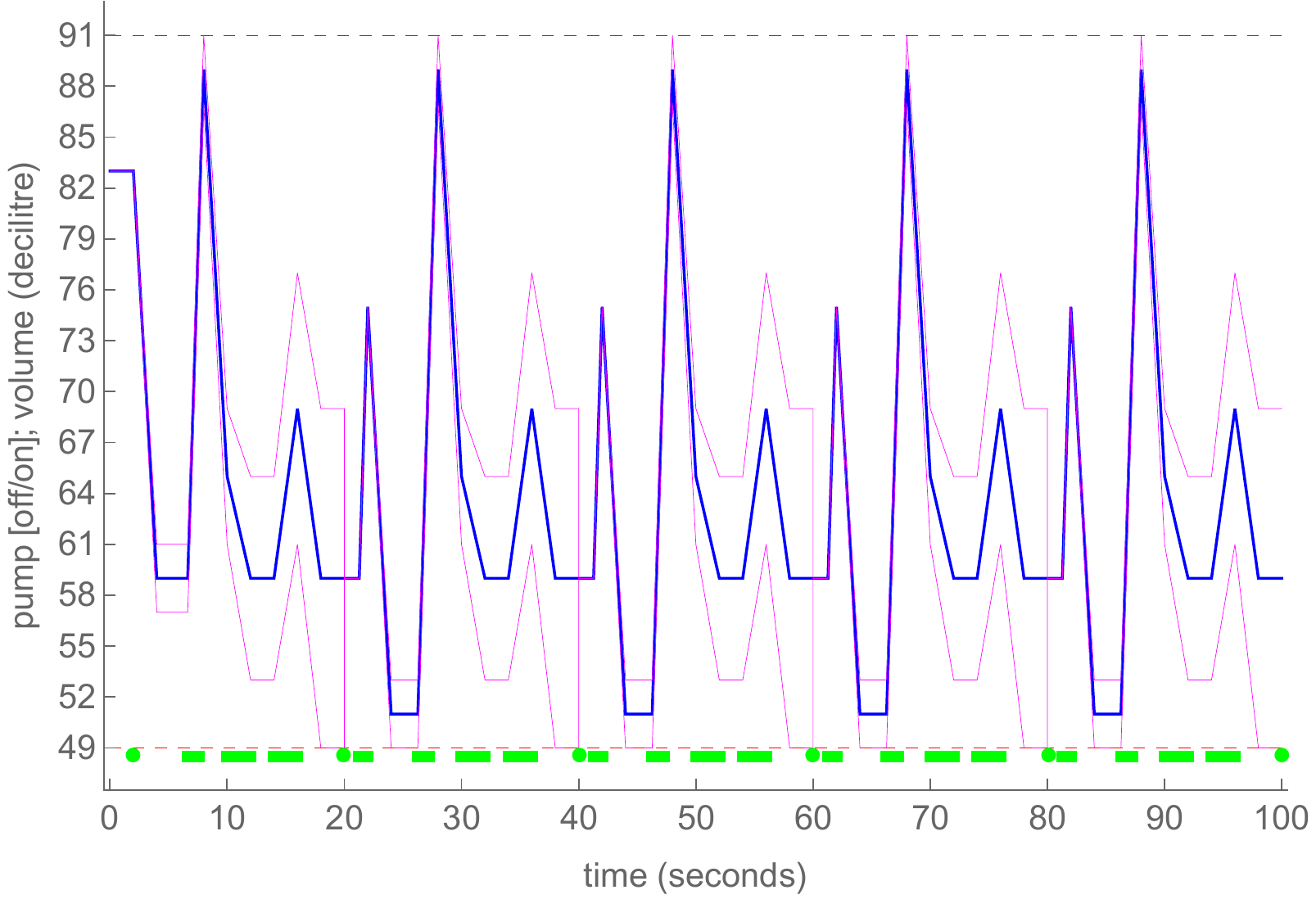}
\end{subfigure}
 \caption{Simulations of 5 consecutive machine cycles
   for~$\calH_1(\epsilon)$ and~$\calH_2(\epsilon)$.}
 \label{fig:simulations}
\end{figure}

We proceed by comparing the performances of our strategies in terms of
accumulated oil volume.  Fig.~\ref{fig:simulations} shows the
result of simulating our strategies for a duration of~$100\;\si{\second}$.
The plots illustrate in blue (resp.~red) the dynamics of the mean (resp.~min/max) oil level in the accumulator 
as well as the state of the pump. 
The~initial volume used for the simulations is
$8.3\litr$, as done in~\cite{CJLRR09} for evaluating respectively the
Bang-Bang controller, the Smart Controller developed by HYDAC, and the
controllers G1M1 and G2M1 synthesized with
\textsc{uppaal-tiga}. 

\begin{table}[t]
\centering
\setlength{\tabcolsep}{1ex}
\begin{tabular}[t]{|c|c|c|}
\hline
Controller & Acc.\ vol. (\si{\litre}) & Mean vol. (\si{\litre}) \\ \hline
\hline
$\mathcal{H}_1$ & 1081.77 & 5.41 \\ \hline
$\mathcal{H}_2$ & 1158.90 & 5.79 \\ \hline
$\mathcal{H}_1(\epsilon)$ & 1200.21 & 6.00 \\ \hline
$\mathcal{H}_2(\epsilon)$ & 1323.42 & 6.62 \\ \hline 
\end{tabular}
\begin{tabular}[t]{|c|c|c|}
\hline
Controller & Acc.\ vol. (\si{\litre}) & Mean vol. (\si{\litre}) \\ \hline
\hline
Bang-Bang & 2689 & 13.45 \\ \hline
\textsc{hydac} & 2232 & 11.60 \\ \hline
G1M1 & 1518 & 7.59 \\ \hline
G2M1 & 1489 & 7.44 \\ \hline
\end{tabular}
\medskip
\caption{Performance based on simulations of $200\;\si{\second}$ starting with $8.3\litr$.}
\label{tab:simulations}
\end{table}

Table~\ref{tab:simulations} presents, for each of the strategies, the
resulting accumulated volume of oil, and the corresponding mean
volume.  There is a clear evidence that the strategies for
  $\mathcal{H}_1$ and~$\mathcal{H}_2$ outperform all the other
  strategies. Clearly, this is due to the fact that they assume full
  precision in the rates, and allow for more switches of the pump.
  However, these results shall be read as what
  one could achieve by investing in more precise
  equipment.
%
The~results also confirm that both our strategies outperform those
presented in~\cite{CJLRR09}. In particular the strategy for
$\mathcal{H}_1(\epsilon)$ provides an improvement of $55\%$, $46\%$,
$20\%$, and $19\%$ respectively for the Bang-Bang controller, the
Smart Controller of HYDAC, and the two strategies synthesized with
\textsc{uppaal-tiga}.

\subsubsection{Tool chain\protect\footnote{More details on our scripts are available at \url{http://people.cs.aau.dk/~giovbacci/tools.html}, together with the models we used for our examples and case study.}.}
Our results have been obtained using Mathematica and
Mjollnir. Specifically, Mathematica was used to construct
the formulas modelling the post-fixpoint of the energy functions, calling Mjollnir
for performing quantifier elimination on~them.  The~combination
of both tools allowed us to solve one of our formulas with 27 variables
in a compositional manner in ca.~$20\;\si{\milli\second}$, while
Mjollnir alone would take more than 20 minutes. 
Mjollnir was preferred to Mathematica's built-in support for quantifier elimination 
because the latter does not scale.



%% file: conclu.tex
We developed a novel framework allowing  for the synthesis of safe and
optimal  controllers, based  on  energy timed  automata. Our  approach
consists in a translation to first-order linear arithmetic expressions
representing our  control problem, and solving  these using quantifier
elimination and simplification.  We~demonstrated the applicability and
performance of  our approach  by revisiting the  HYDAC case  study and
improving its best-known solutions.

Future   works  include   extending   our  results   to  non-flat   and
non-segmented  energy  timed   automata.   However,  existing  results
\cite{Mar11}  indicate   that  we  are   close  to  the   boundary  of
decidability.  Another interesting continuation  of this work would be
to            add           {\sc            Uppaal           Stratego}
\cite{DBLP:conf/atva/DavidJLLLST14,DBLP:conf/tacas/DavidJLMT15} to our
tool chain.   This would allow  to optimize the  permissive strategies
that  we  compute  with  quantifier  elimination  in  the  setting  of
probabilistic uncertainty, thus obtaining controllers that are optimal
with respect to expected accumulated oil volume.


%% file: appendix.tex
\section{Proof of Theorem~\ref{thm-ecir}}\label{app-ecir}

\input{app-eta}

\section{Proof of Theorem~\ref{thm-optU}}\label{app-optU}

\thmoptU*

\input{optU}

\section{Proof of Theorem~\ref{unc-thm}}

\input{app-etau}

\section{Details on the HYDAC case study}
\input{hydac}

%% file: app-eta.tex
\subsubsection{Binary energy relations.}

Let $\calP=\tuple{\{s_i \mid 0 \le i \le n\},\{s_0\},X,I,r,T}$ be an
ETP from $s_0$ to~$s_n$.
Let~$E\subseteq \calI(\bbQ)$ be an energy constraint. The~\emph{binary
  energy relation}~$\calR_\calP^{E} \subseteq E \times E$ for~$\calP$
under energy constraint~$E$ relates all pairs~$(w_0,w_1)$ for which
there is a finite run of~$\calP$ from~$(s_0,\mathbf{0},w_0)$
to~$(s_n,\mathbf{0},w_1)$ satisfying energy constraint~$E$.
This~relation is characterized by the following first-order
formula:
\begin{xalignat*}1
  \calR_\calP^E(w_0,w_1) \iff \exists (d_i)_{0\leq i<n} .\
  & \Phi_{\text{timing}} \wedge \Phi_{\text{energy}} \wedge
    w_1=w_0+\sum_{k=0}^{n-1}(d_k\cdot r(s_k) + u_k) 
\end{xalignat*}
where $\Phi_{\text{timing}}$ encodes all the timing constraints that
the sequence~$(d_i)_{0\leq i<n}$ has to fulfill, while
$\Phi_{\text{energy}}$ encodes the energy constraints. More precisely:
\begin{itemize}
\item timing constraints are obtained by computing the clock
  valuations in each state of the execution, and expressing that those
  values must satisfy the corresponding invariants and guards. The
  value of a clock in a state is the sum of the delays~$d_j$ since the
  last reset of that clock along the ETP.
\item energy constraints are obtained by expressing the value of the
  energy level in each state as the sum of the initial energy level,
  the energy $r(s_i)\cdot d_i$ gained or consumed in each intermediary
  state, and the updates~$u_i$ of the transitions that have been
  traversed. All those values are constrained to lie in~$E$.
\end{itemize}

It is easily shown that $\calR_P^E$ is a closed, convex subset of
$E\times E$ (remember that we consider closed clock constraints); thus
it~can be described as a conjunction of a finite set of linear
constraints over $w_0$ and~$w_1$ (with non-strict inequalities), using
quantifier elimination of variables~$(d_i)_{0\leq i<n}$.


\subsubsection{Energy functions.}

We~now focus on properties of energy relations. First notice that
for any interval~$E\in\calI(\bbQ)$, the partially-ordered set
$(\calI(E),\supseteq)$ is $\omega$-complete, meaning that for any chain
$(I_j)_{j\in\bbN}$, with $I_j\supseteq I_{j+1}$ for all~$j$, the limit
$\bigcap_{j\in\bbN} I_j$ also belongs to~$\calI(E)$.
By~Cantor's Intersection Theorem, if additionally each interval~$I_j$
is non-empty, then so is the limit $\bigcap_{j\in\bbN}
I_j$. 

With an energy relation~$\calR_\calP^E$, we~associate an
\emph{energy function} (also denoted with~$\calR_\calP^E$, or
simply~$\calR$, as long as no ambiguity may arise), defined for any
closed subinterval~$I\in \calI(E)$ as
\[
\calR(I)=\{w_1 \in E\mid \exists w_0\in I.\ \calR(w_0,w_1)\}.
\]
Symmetrically, we~let
\[
\calR^{-1}(I)=\{w_0\in E \mid \exists w_1\in I.\ \calR(w_0,w_1)\}.
\]
Observe that $\calR(I)$ and $\calR^{-1}(I)$ also belong
to~$\calI(E)$ (because relation~$\calR$ is closed and convex).
Moreover, $\calR$ and~$\calR^{-1}$ are monotonic: 
  for any two intervals $I$ and~$J$ in~$\calI(E)$ such that $I\subseteq J$, it~holds
  $\calR(I)\subseteq \calR(J)$ and
  $\calR^{-1}(I)\subseteq \calR^{-1}(J)$.

Energy functions~$\calR$ and~$\calR^{-1}$ also satisfy the following
continuity properties:%
\lemcontinuity*

\begin{proof}
  For any~$i\in\bbN$, we~have $I_i \supseteq \bigcap_{j\in\bbN} I_j$.
  By~monotonicity of~$\calR^{-1}$, we~get $\calR^{-1}(I_i)\supseteq
  \calR^{-1}(\bigcap_{j\in\bbN}I_j)$. It~follows $\bigcap_{i\in\bbN}
  \calR^{-1}(I_i) \supseteq \calR^{-1}(\bigcap_{j\in\bbN}I_j)$.

  Now, let $w_0\in\bigcap_{j\in\bbN} \calR^{-1}(I_j)$. Then for
  all~$i\in\bbN$, there exists $w_1^i$ such that
  $\calR(w_0,w_1^i)$. It~follows that for any~$i\in\bbN$,
  $\calR(\{w_0\}) \cap I_i$ is a non-empty interval of~$\calI(E)$.
  Applying Cantor's Intersection Theorem, we~get that
  $\bigcap_{i\in\bbN} \calR(\{w_0\}) \cap I_i$ is a non-empty interval
  of~$\calI(E)$. This~intersection can be rewritten as $\calR(\{w_0\})
  \cap \bigcap_{i\in\bbN} I_i$; hence there exists $w_1\in
  \bigcap_{i\in\bbN} I_i$ such that $\calR(w_0,w_1)$, which proves
  that $w_0\in \calR^{-1}(\bigcap_{i\in\bbN} I_i)$. \qed
\end{proof}

\subsubsection{Composition and fixpoints of energy functions.}



Consider a finite sequence of paths~$(\calP_i)_{1\leq i\leq k}$.
Clearly, the energy relation for this sequence can be obtained as the
composition of the individual energy relations $\calR_{\calP_k}^E
\circ \cdots \circ \calR_{\calP_1}^E$; the~resulting energy relation
still is a closed convex subset of~$E\times E$ that can be described
as the conjunction of finitely many linear constraints over~$w_0$
and~$w_1$.  As~a special case, we~write $(\calR_\calP^E)^k$ for the
composition of $k$ copies of the same relations~$\calR_\calP^E$.

Now, using Lemma~\ref{lemma-continuity}, we~easily prove that the
greatest fixpoint $\nu \calR^{-1}$ of~$\calR^{-1}$ in the complete
lattice $(\calI(E), \supseteq)$ exists and equals:
\[
\nu \calR^{-1} = \bigcap_{i\in\bbN} (\calR^{-1})^i(E).
\]
Moreover $\nu \calR^{-1}$ is a closed (possibly empty) interval.  Note
that $\nu \calR^{-1}$ is the maximum subset $S_\calR$ of~$E$ such
that, starting with any $w_0\in S_\calR$, it~is possible to
iterate~$\calR$ infinitely many times (that~is, for~any~$w_0\in
S_\calR$, there exists $w_1\in S_\calR$ such that
$\calR(w_0,w_1)$---any~such set $S$ is a post-fixpoint of $\calR^{-1}$ in
the sense that $S\subseteq \calR^{-1}(S)$).

In~the~end, if $\calR$ is the energy relation of a cycle~$\calC$ in
the SETA, then $\nu \calR^{-1}$ precisely describes the set of initial
energy levels allowing infinite runs through~$\calC$ satisfying the
energy constraint~$E$.

Now if $\calR$ is the energy relation for a cycle~$\calC$, described as
the conjunction~$\phi_\calC$ of a finite set of linear constraints,
we~can characterize those intervals $[a,b]\subseteq E$ that constitute
a post-fixpoint for~$\calR^{-1}$ by the following first-order formula:
\[
a \leq b \wedge a\in E \wedge b\in E \wedge \forall w_0\in[a;b].\
\exists w_1\in[a;b].\ \phi_\calC(w_0,w_1).
\]

Applying quantifier eliminination (to~$w_0$ and~$w_1$), the~above
formula may be transformed into a direct constraint on~$a$ and~$b$,
characterizing all post-fixpoints of~$\calR^{-1}$.  We~get a
characterization of~$\nu \calR^{-1}$ by computing the values of $a$ and~$b$
that satisfy these constraint and maximize~$b-a$.

\subsubsection{Algorithm for flat segmented energy timed automata.}

\input{app-eta-algo}

%% file: app-eta-algo.tex
Following Example~\ref{ex-iel-seta}, we~now prove that we can solve
the energy-constrained infinite-run problem for any flat SETA.  The
next theorem is crucial for our algorithm:


\thmterm*

\begin{proof}
Assume that $I\cap\nu \calR^{-1} = \emptyset$. Then:
\begin{align*}
\emptyset &= I\cap\nu \calR^{-1} \\
  &= \textstyle I\cap \bigcap_{n \in \bbN} \bigl(\calR^{-1}\bigr)^{n}(E) \\
  &= \textstyle I\cap \bigcap_{n \in \bbN} \bigl(\calR^{n}\bigr)^{-1}(E) 
     \tag{by $\calR^{-1} \circ \calR^{-1} = (\calR \circ \calR)^{-1}$} \\
  &= \textstyle \bigcap_{n \in \bbN} \bigl( I \cap \bigl(\calR^{n}\bigr)^{-1}(E)
     \bigr)
\end{align*}
Note that $\big(I \cap (\calR^{n})^{-1}(E)\big)_{n \in \bbN}$ is a
decreasing sequence because
$\big((\calR^{-1})^{n}(E)\big)_{n \in \bbN}$~is. Therefore, by
Cantor's intersection theorem $I \cap \big(\calR^{n}\big)^{-1}(E)
= \emptyset$ for some $n \in \bbN$. But only elements
$w_0 \in \big(\calR^{n}\big)^{-1}(E)$ admit some $w_1 \in E$ such that
$\calR^{n}(w_0,w_1)$. Therefore $\calR^n(I) = \emptyset$.
\qed
\end{proof}

We will show that the energy-constrained infinite run problem is
decidable for flat SETAs.  The decision procedure traverses
the underlying graph of~$\calA$, forward propagating an initial energy
interval $I_0 \subseteq E$ looking for a simple cycle~$C$ such that
$\nu \calR_C^{-1} \cap I \neq \emptyset$, where $I \subseteq E$ is
the energy interval forward-propagated until reaching the cycle.

\begin{algorithm}[tb]
    \algsetup{linenodelimiter=.}
    \caption{Infinite Run}
    \begin{algorithmic}[1]  
    \REQUIRE  A \emph{flat SETA} $\calA=\tuple{S,T,P}$; initial state $m_0 \in S$;
    energy interval $I_0$
    \STATE $W \gets \{(m_0, I_0, c)\}$ \COMMENT{initialise the waiting list}\label{line:1}
    \WHILE{$W \neq \emptyset$} 
    	\STATE pick $\tuple{m,I,\flag} \in W$ 	\COMMENT{pick an element from the waiting list}\label{line:3}
	\STATE $W \gets W \setminus \tuple{m,I,\flag}$ \COMMENT{remove the element from the waiting list}
	\IF[the node $m$ shall be explored without following a cycle]{$\flag = \bar{c}$}
		\FOR{\textbf{each} $(m,m') \in T$ that is not part of a simple cycle of $(S,T)$}\label{line:6}
			\STATE $W \gets W \cup \{ \tuple{m', \calR_{P(m,m')}^E(I), c}\}$ \COMMENT{add this new task to the waiting list}\label{line:7}
		\ENDFOR
	\ELSE[the node $m$ shall be explored by following a cycle]
		\IF{$m$ belongs to a cycle of $(S,T)$}
			\STATE let $\calC = (m_1,m_2)\cdots (m_k,m_{k+1})$ be the simple cycle s.t.\ $m = m_1 = m_{k+1}$
			\STATE let $\calR_\calC = \calR_{P(m_k,m_{k+1})} \circ \cdots \circ \calR_{P(m_1,m_{2})}$
			\COMMENT{energy relation of the cycle}
        			\IF[check if there is an infinite run via the cycle $C$]{$I \cap \nu \calR_\calC^{-1} \neq \emptyset$}
        				\RETURN $\mathtt{tt}$
        			\ELSE[the cycle can be executed only finitely many times]
				\STATE $W \gets W \cup \{(m, I, \bar{c})\}$ \COMMENT{add a new task to the waiting list}\label{line:16}
        				\STATE $i \gets 0$ \COMMENT{initialise the number of cycle executions}
        				\WHILE[while $i$-th energy relation is satified]{$\calR_\calC^i(I) \neq \emptyset$}\label{line:18}
        					\FOR{$1 \leq j < k$} 
						\STATE let $\calR_{\calP_j} = \calR_{P(m_j,m_{j+1})} \circ \cdots \circ \calR_{P(m_1,m_{2})}$ \COMMENT{unfold $C$ up to $m_{j+1}$}\label{line:20}
        						\STATE $W \gets W \cup \{ \tuple{m_{j+1}, \calR_{\calP_j}(\calR_C^i(I)), \bar{c}}\}$\label{line:21}
						\COMMENT{add a task to the waiting list}
        					\ENDFOR
        					\STATE $i \gets i + 1$ \COMMENT{increment the number of cycle executions}
        				\ENDWHILE\label{line:24}
        			\ENDIF
		\ELSE[$m$ doesn't belong to a cycle]
			\STATE $W \gets W \cup \{(m, I, \bar{c})\}$ \COMMENT{add a new task to the waiting list}\label{line:27}
		\ENDIF
	\ENDIF
     \ENDWHILE
     \RETURN $\mathtt{ff}$ \COMMENT{no infinite run could be found}
    \end{algorithmic}
    \label{algapp:infiniterun}\label{algapp:infiniteruns}
\end{algorithm}
Algorithm~\ref{algapp:infiniterun} gives a detailed description of the
decision procedure. The~procedure traverses the underlying graph of
the flat SETA~$\calA$, namely $G = (S,T)$, using a waiting list~$W$ to
keep track of the \node s that need to be further
explored. The~list~$W$ contains tasks of the form $\tuple{m,I,\flag}$
where $m \in S$ is the current \node, ${I \in\calI(E)}$~is the current
energy interval, and $\flag \in \{c, \bar{c}\}$ is a flag
indicating if $m$~shall be explored by following a cycle it belongs to
($\flag = c$), or by skipping that cycle ($\flag
= \bar{c}$). The~algorithm first initialises the waiting list
with the initial task (cf.~line~\ref{line:1}).

The~main \textbf{while} loop
processes each task in the waiting list, as long as the list is not empty.
It~picks a task $\tuple{m,I,\flag}$ from~$W$ (line~\ref{line:3}).
If~$\flag = \bar{c}$, the~exploration will continue from \node s~$m'$
adjacent to~$m$ by forward propagating the current energy interval~$I$
following the timed path~$P(m,m')$
(cf.~lines~\ref{line:6}-\ref{line:7}).  Note that the choice of the
arcs~$(m,m')$ ensures that $m'$~does not belong to the same cycle
as~$m$, thus skipping the cycle with~$m$.

Otherwise, if $\flag = c$, the exploration tries to follow the simple
cycle that contains~$m$. If~$m$~does not belong to any cycle the
current task will be simply put back in the waiting list with the
opposite flag (cf.~line~\ref{line:27}). In~case $m$~belongs to the
simple cycle $\calC = (m_1,m_2)\cdots (m_k,m_{k+1})$, the~energy
relation~$\calR_\calC^E$ is used to check if for the current energy interval
there exists an infinite run along the cycle~$\calC$. If~such is not the
case, the~cycle will be iterated only finitely many times (cf.~lines~\ref{line:16}-\ref{line:24}).
This is done by adding in~$W$ the current task with the flag set
to~$\bar{c}$---corresponding to zero executions of the cycle---then
for each execution~$i$ of~$\calC$, the~cycle is unfolded up its $j$-th
transition and the task
$\tuple{m_{j+1}, \calR_{\calP_j}^E((\calR_C^E)^i(I)), \bar{c}}$ is
added to the waiting list---corresponding to $i$~executions of~$\calC$
followed by a tail $(m_1,m_2) \cdots (m_j,m_{j+1})$. 
Termination of the \textbf{while} loop in lines~\ref{line:18}-\ref{line:24} is
ensured by Theorem~\ref{th:termination}, and by the flatness
assumption on~$\calA$, which ensures that each node $m \in S$ belongs to at most
one (simple) cycle, so~that once the execution has left the cycle
where $m$ belongs~to, the exploration won't visit~$m$~again.


\thmecir*

\begin{figure}[t]
\definecolor{invariant}{rgb}{0.0, 0.5, 0.0}%
\definecolor{guard}{rgb}{0.36, 0.54, 0.66}%
\definecolor{reset}{rgb}{0.57, 0.36, 0.51}%
\centering
\begin{tabular}{c}
\begin{tikzpicture}[location/.style={circle, draw=gray!90, thick}]
\draw 
  (0,0) node[location] (l0) {$s_0$}
  	node[font=\small,left=5mm] {$(S,T)\colon{}$}
  ($(l0)+(1,-0.5)$) node[location] (l3) {$s_1$}
  ($(l0)+(right:2.5)$) node[location] (l6) {$s_2$}
  ;
\path[-latex, font=\small]
	(l0) edge[bend left] (l6)
	(l6) edge[loop right] (l6)
	(l0) edge (l3)
	(l3) edge[loop right] (l3)
	(l3) edge (l6)
        (l0)+(-.5,0.5) edge (l0);
\end{tikzpicture}
\\[4ex]
  \begin{tikzpicture} 
    \draw (0,0) node[rond6,fill=red!40!white,font=\small] (a) {$s_1$}
      node[font=\scriptsize,below=2.5mm] {$\rate\colon \!\!-1$} 
      node[font=\small,left=3mm] {$P_{1,1}$:};
    \draw ($(a)+(right:1.5)$) node[rond6,vert,font=\small] (b) {$s_4$}
      node[font=\scriptsize,below=2.5mm] {$\rate\colon 3$};
    \draw ($(b)+(right:2.3)$) node[rond6,fill=red!40!white,font=\small] (c) {$s_1$};
    \draw[font=\scriptsize] 
      (a) edge[-latex'] node[pos=0.65,above] {$\update:\!\!-3$} (b)
      (b) edge[-latex'] node[near end,above] {$\update\colon  \!\!-1$} 
            node[near start,above] {$x = 1$}
            node[below] {$x:=0$} (c);
  \end{tikzpicture}
\end{tabular}
\quad
\begin{tabular}{l}
  \begin{tikzpicture} 
    \draw (0,0) node[rond6,fill=red!40!white,font=\small] (a) {$s_0$}
      node[font=\scriptsize,below=2.5mm] {$\rate\colon 0$} 
      node[font=\small,left=3mm] {$P_{0,1}$:};
    \draw ($(a)+(right:2.3)$) node[rond6,fill=red!40!white,font=\small] (b) {$s_1$};
    \draw[font=\scriptsize] (a) edge[-latex'] 
      node[near end,above] {$\update\colon 4$}  
      node[near start,above] {$x = 1$}
      node[below] {$x:=0$} (b);
  \end{tikzpicture}
  \\[-1ex]
    \begin{tikzpicture} 
    \draw (0,0) node[rond6,fill=red!40!white,font=\small] (a) {$s_1$}
      node[font=\scriptsize,below=2.5mm] {$\rate\colon \!\!-1$} 
      node[font=\small,left=3mm] {$P_{1,2}$:};
    \draw ($(a)+(right:2.3)$) node[rond6,fill=red!40!white,font=\small] (b) {$s_2$};
    \draw[font=\scriptsize] (a) edge[-latex'] 
      node[near end,above] {$\update\colon 2$}  
      node[near start,above] {$x = 1$}
      node[below] {$x:=0$} (b);
  \end{tikzpicture}
  \\[-1ex]
  \begin{tikzpicture} 
    \draw (0,0) node[rond6,fill=red!40!white,font=\small] (a) {$s_0$}
      node[font=\scriptsize,below=2.5mm] {$\rate\colon 0$} 
      node[font=\small,left=3mm] {$P_{0,2}$:};
    \draw ($(a)+(right:1.5)$) node[rond6,vert,font=\small] (b) {$s_5$}
      node[font=\scriptsize,below=2.5mm] {$\rate\colon 5$} ;
    \draw ($(b)+(right:2.3)$) node[rond6,fill=red!40!white,font=\small] (c) {$s_2$};
    \draw[font=\scriptsize] 
      (a) edge[-latex'] node[near end,above] {$\update\colon 4$} (b)
      (b) edge[-latex'] node[near end,above] {$\update\colon \!\!-5$} 
            node[near start,above] {$x = 1$}
            node[below] {$x:=0$} (c);
  \end{tikzpicture}
\end{tabular}
\\
  \begin{tikzpicture} 
    \draw (0,0) node[rond6,fill=red!40!white,font=\small] (a) {$s_2$}
      node[font=\scriptsize,below=2.5mm] {$\rate\colon 2$} 
      node[font=\small,left=3mm] {$P_{2,2}$:};
    \draw ($(a)+(right:1.5)$) node[rond6,vert,font=\small] (b) {$s_6$}
      node[font=\scriptsize,below=2.5mm] {$\rate\colon 5$};
    \draw ($(b)+(right:1.5)$) node[rond6,vert,font=\small] (c) {$s_7$}
      node[font=\scriptsize,below=2.5mm] {$\rate\colon 2$};
    \draw ($(c)+(right:2.3)$) node[rond6,fill=red!40!white,font=\small] (d) {$s_2$};
    \draw[font=\scriptsize] 
      (a) edge[-latex'] node[pos=0.65,above] {$\update:\!\!-3$} (b)
      (b) edge[-latex'] node[near end,above] {$\update\colon 0$} (c)
      (c) edge[-latex'] node[near end,above] {$\update\colon 0$} 
            node[near start,above] {$x = 1$}
            node[below] {$x:=0$} (d);
  \end{tikzpicture}
\caption{An example of SETA $\calA = \protect\tuple{S,T,P}$ with implicit global variant $x \leq 1$. 
The~map~$P$ associates with each $(s_i,s_j) \in T$ the ETP $P_{i,j}$.} \label{fig:etaexample}
\end{figure}
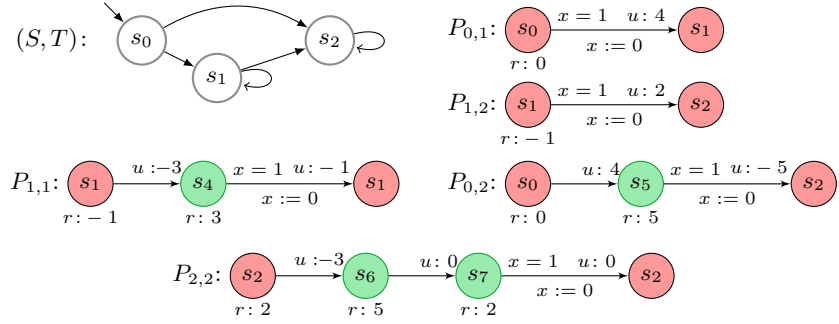

\begin{example}
Consider the SETA $\calA = \tuple{S,T,P}$ depicted in
Fig.~\ref{fig:etaexample} and the energy constraint $E =
[0;6]$. We~describe a step-by-step execution of
Algorithm~\ref{algapp:infiniterun} starting with $s_0 \in S$ and initial
energy interval $I_0 = [0;0]$.

The waiting list is initialised as $W_0 = \{(s_0,I_0,c)\}$. After the
first execution of the main \textbf{while} loop, $W_1
= \{(s_0,I_0,\bar{c})\}$ because $s_0$ does not belong to any simple
cycle of $(S,T)$.
In the second iteration, we~pick the task $(s_0,I_0,\bar{c})$ and we
update the waiting list as $W_2 = \{(s_1,[0;0],c), (s_1,[4;4],c),
(s_2,[0;1],c)\}$. In~the third iteration, we~pick the task
$(s_2,[0;1],c)$ from~$W_2$. Since~$s_2$~belongs to the self-cycle
$\calC = (s_2,s_2)$, we~compute $[0;1] \cap \nu \calR^{-1}_{\calC} =
[0;1] \cap [\frac{5}{3};6] = \emptyset$. Thus, we proceed by computing
$\calR^0([0;1]) = [0;1]$, $\calR^1([0;1]) = [0;0]$ and $\calR^2([0;1])
= \emptyset$, and update the waiting list as $W_3 = \big(W_2 \setminus
(s_2,[0;1],c)\big) \cup \{(s_2,[0;1],\bar{c}),
(s_2,[0;0],\bar{c})\}$. In~the~fourth and fifth iterations, we~pick
the tasks $(s_2,[0;1],\bar{c})$ and $(s_2,[0;0],\bar{c})$,
respectively. Since~$s_2$~cannot escape from the the self-cycle, we
will not insert any tasks in the waiting list, thus having $W_5
= \{(s_1,[0;0],c), (s_1,[4;4],c)\}$. During the sixth iteration, we
pick the task $(s_1,[4;4],c)$. Since~$s_1$ belongs to the self-cycle
$\calC' = (s_1,s_1)$, we compute $[4;4] \cap \nu\calR^{-1}_{\calC'} =
[4;4] \cap \emptyset = \emptyset$. Thus we proceed by computing
$\calR^0([4;4]) = [4;4]$, $\calR^1([4;4]) = [0;3]$, $\calR^2([4;4]) =
[2;2]$, and $\calR^3([4;4]) = \emptyset$ and obtaining $W_6
= \bigl(W_5 \setminus (s_1,[4;4],c)\bigr) \cup \{(s_1,[4;4],\bar{c}),
(s_1,[4;4],\bar{c}), (s_1,[2;2],\bar{c})\}$. In~the seventh iteration,
we~pick the task $(s_1,[4;4],\bar{c})$. The~only transition that
escapes from the self-cycle of~$s_1$ is~$(s_1,s_2)$, thus we get $W_7
= \bigl(W_6 \setminus (s_1,[4;4],\bar{c})\bigr)\cap \{ (s_2,
[5;5],c) \}$. Finally, we~pick the task $(s_2, [5;5],c)$ and since
$[5;5] \cap \nu \calR^{-1}_{\calC''} = [5;5] \cap
[\frac{5}{3};6] \neq \emptyset$ where $\calC'' = (s_2,s_2)$, we~stop the
computation and return $\mathtt{tt}$.
\end{example}

%% file: optU.tex





\begin{proof}
  Let $\calA$ be a flat SETA and $L \in \bbQ$ be the fixed lower
  bound.

  Let $\calC$ be a simple cycle of $\calA$ (which may formally be the
  concatenation of several energy timed paths but w.l.o.g. we can
  assume it is a single energy timed path).  We~analyze when this
  cycle can be iterated, and for which upper bound~$U$.  Adding $U$ as
  a parameter, we can refine the approach
  of Section~\ref{sec-eta}, and safely define the
  ternary energy relation $\calR_\calC(w_0,w_1,U)$ as
  $\calR_\calC^{[L;U]}(w_0,w_1)$. It~is a convex subset of~$\bbR^3$,
  described as a conjunction of a finite set of linear constraints
  over~$w_0$, $w_1$ and~$U$ (with non-strict inequalities and rational
  coefficients).  We~can then define the predicate
  $\calR_\calC^\infty(a,b,U)$ as:
  \[
  \calR_\calC^\infty(a,b,U) \iff L \le a \le b \le U \wedge \forall
  w_0 \in [a;b],\ \exists w_1 \in [a;b].\ \calR_\calC(w_0,w_1,U)
  \]
  characterizing the intervals~$[a;b]$ and upper-bounds~$U$ such
  that~$\calC$ can be iterated infinitely many times from any initial
  value in~$[a;b]$ with energy constraint~$[L;U]$.  This~relation is
  again a convex subset of~$\bbR^3$, described as a conjunction of a
  finite set of linear constraints over~$a$, $b$ and~$U$ (with
  non-strict inequalities and rational coefficients).

  For a fixed $U \in \bbQ$, this predicate coincides with the greatest
  fixpoint $\nu (\calR_{\calC}^{[L;U]})^{-1}$ that was
  discussed on page~\pageref{page-subsec:composition_fixpoints}. Hence
  $\calR_\calC^\infty(a,b,U)$ holds if and only if for every $w_0 \in
  [a;b]$, there is an infinite run starting at $(s_0,\mathbf{0},w_0)$
  (where $s_0$ is the first state of $\calC$) satisfying the energy
  constraint $[L;U]$. Furthermore,
  since $a\geq L$ and $L$ is fixed, and because we only have
  non-strict constraints, 
  there is a least value $a^\calC_{\min} \in \bbQ$ such that the set
  $\{(b,U) \mid \calR_\calC^\infty(a^\calC_{\min},b,U)\}$ is non-empty.
  In~particular:%
  \begin{lemma}
    \begin{itemize}
    \item For any energy level $w<a^\calC_{\min}$, and for any~$U$,
      there are no infinite run from $(s_0,\mathbf{0},w)$ cycling
      around~$\calC$ and satisfying energy constraint~$[L;U]$;
    \item For every $w \ge a^\calC_{\min}$, there exist~$U$ and an
      infinite run from $(s_0,\mathbf{0},w)$ cycling around~$\calC$
      and satisfying energy constraint $[L;U]$.
    \end{itemize}
  \end{lemma}
  \begin{proof}
    The first part of the lemma is a direct consequence of the
    analysis of the fixpoint $\nu
    \Big(\calR_{\calC}^{[L;U]}\Big)^{-1}$ made in
    Sec.~\ref{sec-eta}. 

    For the second property, we first realize that there is $(b,U) \in
    \bbQ^2$ such that $\calR^\infty_\calC(a_{\min}^\calC,b,U)$, which
    means in particular that there is an infinite run from
    $(s_0,\mathbf{0},a_{\min}^\calC)$ cycling around $\calC$ and
    satisfying the energy constraint $[L;U]$. Now, by mimicking the
    same delays, it is easy to get that for every $w \ge
    a_{\min}^\calC$, there is an infinite run from
    $(s_0,\mathbf{0},w)$ satisfying the energy constraint
    $[L;U+w-a_{\min}^\calC]$. \qed
  \end{proof}

  Coming back to our automaton $\calA$: if there is a solution to the
  energy-constrained infinite-run problem in $\calA$ for some upper
  bound $U$, the witness infinite run must end up cycling in one of
  the cycles of $\calA$. Let $\calC$ be a cycle. We know from the
  lemma above that, to be able to generate a witness infinite run
  cycling around $\calC$, one needs to be able to reach the start of
  that cycle with at least energy level $a_{\min}^\calC$.  Note that
  if we find a finite run reaching the start of cycle~$\calC$ with
  energy level $w \ge a_{\min}^{\calC}$ and satisfying the energy
  constraint $[L;+\infty)$ (only a lower bound constraint) along the
  way, then for some~$U'$ this finite path satisfies the energy
  constaint~$[L;U']$; the concatenation of that finite run with a
  witness infinite run cycling along $\calC$ while satisfying some
  $[L;U]$-energy constraint gives a witness infinite run for the
  existence of an upper bound (with upper bound $\max(U;U')$).

  We therefore study finite runs leading to the start of
  cycle~$\calC$, with only the lower bound~$L$ on the energy
  level. Recall that this problem is in general not easy to
  solve~\cite{BLM14}, and only single-clock automata can be handled in
  general~\cite{BFLM10}. However in the special setting of flat SETA,
  we are able to decide the existence of a well-adapted finite run
  reaching the start of cycle~$\calC$.  Let~$\calP$ be an energy timed
  path. Following a similar approach to the approach developed on
  page~\pageref{page-subsec:binary_energy_relations}, one~can easily define
  a predicate~$\calS_\calP(w_0,w_1)$ that is true whenever there is a
  run satisfying the energy constraint~$[L;+\infty)$, starting with
  energy level~$w_0$ and ending with energy level~$w_1$. From that
  predicate, one can derive the predicates
  $\calS_\calP^{\uparrow}(w_0)$ (resp. $\calS_\calP^{=}(w_0)$,
  $\calS_\calP^{\times}(w_0)$) such that:
  \begin{itemize}
  \item $\calS_\calP^{\uparrow}(w_0) \iff \exists w_1>w_0\
    \text{s.t.}\ \calS_\calP(w_0,w_1)$;
  \item $\calS_\calP^{=}(w_0) \iff \calS_\calP(w_0,w_0)\ \text{and}\
    \neg \calS_\calP^{\uparrow}(w_0)$;
  \item $\calS_\calP^{\times}(w_0) \iff \forall w_1 \ge w_0,\ \neg
    \calS_\calP(w_0,w_1)$.
  \end{itemize}
  In the two first cases, and only in these cases, the path can be
  iterated while satisfying the energy constraint $[L;+\infty)$. In
  the first case, by iterating the path, one can increase the energy
  level up to an arbitrarily high value. In~the second case, only
  energy level~$w_0$ can be reached. These properties are very easy to
  prove (since there is no upper bound), and are therefore omitted.

  Let $\calA$ be a SETA with
  initial energy level~$w_0$. We~perform the following (partial)
  labelling~$\lambda$ of the graph in a forward manner:
  \begin{itemize}
  \item we label the initial macro-state~$m_0$ with $\lambda(m_0) =
    \top$ if there is a path $\calP$ from~$m_0$ to itself, where
    $\calS_\calP^\uparrow(e_0)$ holds; Otherwise we set $\lambda(m_0)
    = w_0$.
  \item let $m$ be a macro-state which does not belong to a cycle, and
    such that all its predecessors have been already labelled
    with~$\lambda$. Write~$(m_i)_{1 \le i \le p}$ for a non-empty list
    of its predecessors, with redundancies if there are multiple
    transitions between macro-states. For~each~$1\le i \le p$,
    write~$\calP_i$ for the ETP labelling the edge~$(m_i,m)$. If~there
    is some~$i$ such that $\lambda(m_i) =\top$, then set $\lambda(m) =
    \top$. Otherwise, define~$w'_i$ for the largest energy level such
    that $\calS_{\calP_i}(w_i,w'_i)$ holds ($w'_i$~can be equal
    to~$+\infty$ whenever $w'_i$ can be made arbitrarily
    large). If~there is a cycle~$\calC$ starting at~$m_i$ such that
    $\calS_\calC^\uparrow(w'_i)$, then set $\lambda(m) =
    \top$. If~$w'_i = +\infty$ for some~$i$, then set $\lambda(m) =
    \top$, otherwise set $\lambda(m) = \max_{1 \le i \le p} w'_i$.
  \end{itemize}

  The following lemma concludes the decidability proof for the
  existence of an upper bound.
  \begin{lemma}
    There is a solution to the upper-bound existence problem if,
    and only~if, there is a cycle $\calC$ starting at some macro-state
    $m$ in $\calA$ such that $a_{\min}^{\calC}$ is well-defined, and
    such that $\lambda(m) = \top$ or $\lambda(m) \ge
    a_{\min}^{\calC}$.
  \end{lemma}

  \begin{proof}
    We can prove the following invariant to the labelling algorithm:
    \begin{itemize}
    \item $\lambda(m) = \top$ if and only if for every $\alpha \in
      \bbR$ there is $w \ge \alpha$ such that energy level $w$ can be
      achieved when reaching $m$;
    \item $\lambda(m) = \alpha$ if, and only~if, $\alpha$ is the maximal
      energy level that can be reached at~$m$.\qed
    \end{itemize}
  \end{proof}

  It remains to discuss the synthesis of the least upper bound for
  which there is a solution to the upper bound synthesis problem. In
  this case, we will restrict to depth-1 flat SETA, that is the graph
  underlying the SETA is a tree, with self-loops at leaves.
  The~general case of flat SETA might be solvable, but we do not have
  a complete proof yet of that general case.  We assume we have found
  a bound~$U$ such that $\calA$ satisfies the infinite path problem
  with energy constraint~$[L;U]$.
  
  Since $\calA$ is depth-1, it can be decomposed as a union of timed
  paths followed by a cycle. Let~$\calP$ be such a path, followed by
  cycle $\calC$. We assume w.l.o.g. that there is an infinite run
  satisfying the energy constraint $[L;U]$ following~$\calP$ and
  cycling along $\calC$. We define the predicate $\calR_{\calP \cdot
    \calC^\omega}(U')$ by
  \[
  U' \le U\ \wedge\ \exists L \le a \le w_1 \le b \le U'\ \text{s.t.}\
  \calR_\calP(w_0,w_1,U') \text{ and } \calR_\calC^\infty(a,b,U')
  \]
  It~is easy to check that $\calR_{\calP \cdot \calC^\omega}(U')$
  holds if and only if $U' \le U$ is a correct upper bound for a
  witness along $\calP \cdot \calC^\omega$. We~can simplify the
  predicate $\calR_{\calP \cdot \calC^\omega}(U')$, and obtain the
  least upper bound as the smallest~$U'$ such that $\calR_{\calP \cdot
    \calC^\omega}(U')$ holds for some~$\calP$ and~$\calC$ in~$\calA$.
  \qed
\end{proof}

%% file: app-etau.tex
The   assumptions   of   perfect   knowledge   of   energy-rates   and
energy-updates  are often  unrealistic, as  is the  case in  the HYDAC
oil-pump  control problem  (see~Section~\ref{sec-hydac}). Rather,  the
knowledge  of energy-rates  and  energy-updates comes  with a  certain
imprecision,  and the  existence of  energy-constrained infinite  runs
must take these into account in  order to be robust.  In~this section,
we~revisit the  energy-constrained infinite-run problem in  the setting
of imprecisions, by~viewing it as a two-player game problem.

\subsubsection{Adding uncertainty to ETA.}

\begin{definition}
  An \emph{energy timed automaton with \textbf{uncertainty}}~(ETAu
  for~short) is a
  tuple~$\calA=\tuple{S,S_0,\Cl,\inv,\rate,T,\epsilon,\Delta}$, where
  $\tuple{S,S_0,\Cl,\inv,\rate,T}$ is an energy timed automaton, with
  $\epsilon\colon S\to \bbQ_{>0}$ assigning imprecisions to rates of
  states and $\Delta\colon T\to\bbQ_{>0}$ assigning imprecisions to
  updates of transitions.
\end{definition}

In the obvious manner, this notion of uncertainty extends to
\emph{energy timed path with uncertainty}~(ETPu) as well~as to
\emph{segmented energy timed automaton with uncertainty}~(SETAu).

Let $\calA=\tuple{S,S_0,\Cl,\inv,\rate,T,\epsilon,\Delta}$ be an ETAu,
and let $\tau = (t_i)_{0 \le i < n}$ be a finite sequence of
transitions, with $t_i = (s_i,g_i,u_i,z_i,s_{i+1})$ for every $i$.
A~finite \emph{run} in $\calA$ on $\tau$ is a sequence of
configurations $\rho=(\ell_j,v_j,w_j)_{0\leq j\leq 2n}$ such that
there exist a sequence of delays~$d=(d_i)_{0\leq i<n}$ for which the
following requirements hold:
\begin{itemize}
\item for all $0\leq j<n$, $\ell_{2j}=\ell_{2j+1}=s_j$, and
  $\ell_{2n}=s_n$;
\item for all $0\leq j<n$, $v_{2j+1}=v_{2j}+d_j$ and
  $v_{2j+2}=v_{2j+1}[z_j\to 0]$;
\item for all $0\leq j<n$, $v_{2j}\models \inv(s_j)$ and
  $v_{2j+1}\models \inv(s_j) \wedge g_j$;
\item for all $0\leq j<n$, it holds that $w_{2j+1}=w_{2j} + d_j\cdot
  \alpha_j$ and $w_{2j+2}=w_{2j+1} + \beta_j$, where
  $\alpha_j\in[r(s_j)-\epsilon(s_j), r(s_j)+\epsilon(s_j)]$ and
  $\beta_j\in[u_j-\Delta(t_j), u_j+\Delta(t_j]$.
\end{itemize}
%
We say that $\rho$ is a possible outcome of $d$ along $\tau$, and that
$w_{2n}$ is a possible final energy level for $d$ along $\tau$, given
initial energy level $w_0$.  Note that in the case of uncertainty, a
given delay sequence $d$ may have several possible outcomes (and
corresponding energy levels) along a given transition sequence $\tau$
due to the uncertainty in rates and updates.  In particular, we~say
that $\tau$ together with $d$ with initial energy level~$w_0$ satisfy
an energy constraint $E\in \calI(\bbQ)$ if any possible outcome run
$\rho$ for $t$ and $d$ starting with~$w_0$ satisfies~$E$. All these
notions are formally extended to ETPu.

Given an ETPu $\calP$, and a delay sequence $d$ for $\calP$ satisfying
a given energy constraint $E$ from initial level $w_0$, we denote by
$\calE^E_{\calP,d}(w_0)$ the set of possible final energy levels.
It~may be seen that $\calE^E_{\calP,d}(w_0)$ is a closed subset
of~$E$.

Now let $\calA=\tuple{S,T,P}$ be an SETAu and let $E$ be an energy
constraint.  A (memoryless\footnote{for the infinite-run problem we
  consider it may be shown that memoryless strategies suffice.})
\emph{strategy} $\sigma$ returns for any macro-configuration $(s,w)$
($s\in S$ and $w\in E$) a pair $(t,d)$, where $t=(s,s')$ is a
successor edge in $T$ and $d\in\bbR+^n$ is a delay sequence for the
corresponding energy timed path, i.e. $n=|P(t)|$.  A (finite or
infinite) execution of $(\rho^i)_i$ writing $\rho^i=(\ell^i_j, x^i_j,
w^i_j)_{0\leq j\leq 2n_i}$, is an outcome of $\sigma$ if the following
conditions hold:
\begin{itemize}
\item $s^i_0$ and $s^i_{2n_i}$ are macro-states of~$\calA$, and
  $\rho^i$ is a possible outcome of $P(s^i_0,s^i_{2n_i})$ for $d$
  where $\sigma(s^i_0,w^i_0)=\big((s^i_0,s^i_{2n_i}),d\big)$;
\item $s^{i+1}_0 = s^i_{2n_i}$ and $w^{i+1}_0=w^i_{2n_i}$.
\end{itemize}
%
Now  we may  formulate  the  infinite-run  problem in  the
setting of uncertainty:

\begin{definition}
  Let $\calA$ be a SETAu, $E\in\calI(\bbQ)$ be an energy constraint,
  and $(s_0,w_0)$ an initial macro-configuration ($s_0$ macro-state of
  $\calA$ and $w_0 \in E$ energy level).  The \emph{energy-constrained
    infinite-run problem} is as follows: does there exist a strategy
  $\sigma$ for $\calA$ such that all runs $(\rho^i)_i$ that are
  outcome of $\sigma$ starting from configuration $(s_0,w_0)$ satisfy
  $E$?
\end{definition}


\subsection{Ternary energy relations}





Let $\calP=(\{s_i \mid 0 \le i \le
n\},\{s_0\},X,I,r,T,\epsilon,\Delta)$ be an ETPu and let
$E\in\calI(\bbQ)$ be an energy constraint.  The ternary energy
relation $\calU^E_\calP\subseteq E\times E\times E$ relates all
triples $(w_0,a,b)$ for which there is a strategy $\sigma$ such that
any outcome of $\rho$ from $(s_0,\mathbf{0},w_0)$ satisfies $E$ and
ends in a configuration $(s_n,\mathbf{0},w_1)$ where $w_1 \in[a;b]$.
This relation can be characterized by the following first-order
formula:
%
\begin{align*}
  \calU_\calP^E(w_0,a,b) \iff &
     \exists (d_i)_{0\leq i<n}. \forall
     (\alpha_i\in[r(s_i)-\epsilon(s_i);r(s_i)+\epsilon(s_i)])_{0\leq i<n}.\\
  &  \forall
    (\beta_i\in[u_j-\Delta(t_j);u_j+\Delta(t_j)])_{0\leq i<n}.\\
    & \Phi_{\text{timing}} \wedge \Phi_{\text{energy}}^u \wedge
      a\leq w_0+\sum_{k=0}^{n-1}(d_k\cdot \alpha_k + \beta_k)\leq b
\end{align*}
%
where $\Phi_{\text{timing}}$ encodes all the timing constraints that
the sequence~$(d_i)_{0\leq i<n}$ has to fulfill and is identical to
that used in the case of full precision.
Also $\Phi_{\text{energy}}^u$ encodes the energy constraints relative
to~$E$.  Formula~$\Phi_{\text{energy}}^u$ is similar to
$\Phi_{\text{energy}}$ from Sec.~\ref{sec-eta},
but refers to $\alpha_i$ and $\beta_i$ rather than to the nominal
rates $r(s_j)$ and updates~$u_i$.

The expression above has two drawbacks: it~mixes existential and
universal quantifiers (which may severely impact efficiency), and the
arithmetic expression is quadratic (for which no efficient tools
provide quantifier elimination). A~better way to characterize the
ternary relation is by expressing inclusion of the set of reachable
energy levels in the energy constraint:
\begin{multline*}
  \calU_\calP^E(w_0,a,b) \iff 
     \exists (d_i)_{0\leq i<n}.\
     \Phi_{\text{timing}} \wedge \Phi_{\text{energy}}^i \wedge {}\\
    \qquad w_0+\sum_{k=0}^{n-1}(r(s_k)\cdot d_k + u_k) +
            \sum_{k=0}^{n-1}([-\epsilon(s_k);\epsilon(s_k)]\cdot d_k
            + [-\Delta(t_k); \Delta(t_k)]) \subseteq [a;b]
\end{multline*}
where $ \Phi_{\text{energy}}^i$ encodes the energy constraints as the
inclusion of the interval of reachable energy levels in the energy
constraint (in~the same way as we do on the second line of the
formula). Interval inclusion can then be expressed as constraints on
the bounds of the intervals.  This~way, we~get linear arithmetic
expressions and no quantifier alternations.
It~is clear  that $\calU^E_{\calP}$ is a closed, convex
subset  of $E\times  E\times  E$  and can  be  described  as a  finite
conjunction  of  linear  constraints  over  $w_0,  a$  and  $b$  using
quantifier elimination.


\subsection{Algorithm for SETAu}

Let $\calA=(S,T,P)$ be a SETAu and let $E\in\calI(\bbQ)$ be an energy
constraint. Let $\calW\subseteq S\times E$ be the maximal set of
configurations satisfying the following:
\begin{align}
  (s,w)\in\calW \,\Rightarrow & \,\exists t=(s,s')\in T. \exists a,b \in
                       E. \nonumber\\
  & \,\,\,\calU^E_{P(t)}(w,a,b) 
    \wedge \forall w'\in[a;b]. (s',w')\in\calW
\label{Weq-app}
\end{align}

Now $\calW$ is easily shown to characterize the set of configurations
$(s,w)$ that satisfy the energy-constained infinite-run problem.
%
%
Unfortunately this characterization does not readily provide an
algorithm.  For this, we make the following restriction and show that
it leads to decidability of the energy-constrained infinite-run
problem.
\begin{description}
\item[(R)] in any of the ETPu $P(t)$ of $\calA$, on at least one of
  its transitions, some clock $x$ is compared with a postive lower
  bound. Thus, there is an (overall minimal) positive time-duration
  $D$ to complete any $P(t)$ of $\calA$.
\end{description}


\thmuncert*

\begin{proof}
  Under hypothesis {\bf{(R)}}, there is a minimum level of imprecision
  for any transition $t=(s,s')$: whenever $\calU^E_{P(t)}(w,a,b)$ then
  $|b-a|\geq D\cdot\Delta_{\text{min}}$, where $\Delta_{\text{min}}$
  is the minimal imprecision within all ETPu $P(t)$ of $\calA$.  Thus
  if $(s,w)\in\calW$ ``due~to'' some transition $t=(s,s')$, then for
  some interval $[a,b]$ with $|b-a|\geq D\cdot\Delta_{\text{min}}$ all
  configurations $(s',w')$ with $w'\in[a,b]$ must be in $\calW$.  Now
  let $N=\left\lceil
    \frac{|E|}{D\cdot\Delta_{\text{min}}}\right\rceil$.  It follows
  that the subset of $E$ given by $\calW_s=\{w' \mid
  (s',w')\in\calW\}$ may be divided into at most $N$ intervals
  $[a_{s',j},b_{s',j}]$ ($1\leq j\leq N$), each of size at least
  $D\cdot\Delta_{\text{min}}$.  We may therefore rewrite equation
  (\ref{Weq-app}) as the first-order formula:
  \begin{align}
    \bigwedge_{s\in S}\bigwedge_{1\leq j\leq N} & [a_{s,j};b_{s,j}]
    \subseteq E \wedge w_0 \in \bigvee_{1 \le j \le N}
    [a_{s_0,j};b_{s_0,j}] \wedge \forall
    w\in[a_{s,j};b_{s,j}]. \nonumber\\
    & \bigvee_{(s,s')\in T}\big[ \exists a,b\in E.
    \,\,\calU^E_{P(s,s')}(w,a,b)\wedge \bigvee_{1\leq k\leq
      N}([a;b]\subseteq [a_{s',k};b_{s',k}])\big]
      \label{Ueq-app}
    \end{align}
    By quantifier elimination, the above may be rewritten as a boolean
    combination of linear constraints over the variables $a_{s,j},
    b_{s,j}$, and determining the satisfiability of the formula is
    decidable.  \qed
\end{proof}

It is worth noticing that we do \textbf{not} assume flatness of the
model for proving the above theorem. Instead, the minimal-delay
assumption {\bf{(R)}} has to be made.

\subsection{Synthesis of optimal upper bound}

For the (optimal) upper-bound synthesis problem, we have the following
results in the setting of uncertainty.

\thmubuncert*

\begin{proof}
  First, for a cycle ETPu $\calC$ and a lower energy bound $L$, we may
  define a quaternary relation $\calX^L_{\calC}$ on $E$ such that
  $\calX^L_{\calC}(w,a,b,U)$ holds if and only if
  $\calU^{[L;U]}_{\calC}(w,a,b)$.  Clearly $\calX^L_{\calC}$ can be
  described as a first-order formula over linear arithmetic, and by
  quantifier elimination as a finite conjunction of linear constraints
  over $w, a, b$ and $U$.

  Now, since $\calA$ is a depth-1 flat SETAu, we can assume
  w.l.o.g. that $\calA$ consists in a path followed by a cycle that
  one tries to iterate. This is w.l.o.g. since a depth-1 flat SETAu
  can be seen as a finite union of such simple automata. Hence we
  assume $\calA = (S,T,P)$ has two macro states $S = \{s,s'\}$, and
  two macro-transitions $T = \{(s,s'),(s',s')\}$. We let $\calP$ be
  the path $T(s,s')$ and $\calC$ be $T(s',s')$. For any given $U$ it
  suffices to capture the set $\calW_{s'}$ with a single interval
  $[a_{s'};b_{s'}]$ (as in the proof of Thm.~\ref{unc-thm}). We may
  now rewrite the equation (\ref{Ueq-app}) as the first-order formula:
%
  \begin{align*}
    w_0 \in [L;U] \wedge \exists a,b \ge L.\
    \calX_{\calP}^L(w_0,a,b,U) \wedge [a;b] \subseteq [a_{s'};b_{s'}]
    \subseteq [L;U] \wedge  \\
    \forall w \in [a_{s'};b_{s'}].\ \exists a',b' \ge L'.\ 
    \calX_{\calC}^L(w,a',b',U) 
    \wedge [a';b'] \subseteq [a_{s'};b_{s'}]
  \end{align*}

  By quantifier elimination the above may be rewritten as a boolean
  combination of linear constraints over the variables $a_{s'},
  b_{s'}$ and $U$, and determining the satisfiability of the formula
  is decidable.  In addition, using linear programming, we may find
  the minimal value of $U$.  \qed
\end{proof}

%% file: hydac.tex
In this section we present an industrial case study that was provided
by the HYDAC company in the context of a European research project
Quasimodo~\cite{quasimodo}.
The case study 
consists in an on-off control system where the system
to be controlled, depicted in Fig~\ref{fig:hydacoverview}, is composed
of
\begin{enumerate*}[label=(\roman*)]
\item a machine that consumes oil, 
\item an accumulator containing oil and a fixed amount of gas in order to put the oil under pressure, and 
\item a controllable pump which can pump oil in the accumulator. 
\end{enumerate*}
When the system is operating, the machine consumes oil under pressure
out of the accumulator. The level of the oil, and so the pressure
within the accumulator, can be controlled by pumping additional oil in
the accumulator (thereby increasing the gas pressure). The~control
objective is twofold: first the level of oil into the accumulator
(and~so the gas pressure) shall be maintained within a safe interval;
second, at~the end of each operating cycle, the accumulator shall be
in a state that ensures the controllability of the following cycle.
Besides these safety requirements, the~controller should also try to
minimize the oil level in the tank, so as to not damage the system.

\subsection{Modelling the oil pump system.}

In this section we describe the characteristics of each component of
the HYDAC case. Then we model the system as a~SETA.
\begin{trivlist}
\item \emph{The Machine.} The oil consumption of the machine is
  cyclic. One cycle of consumptions, as given by HYDAC, consists of
  $10$ periods of consumption, each having a duration of two seconds,
  as~depicted in Figure~\ref{fig:machinecycle}. Each period is
  described by a rate of consumption~$m_r$ (expressed in litres per
  second).
  The~consumption rate is subject to noise: if the mean consumption
  for a period is $c\;\si{\litre/\second}$ (with~$c \geq 0$) its
  actual value lies within $[\max(0,c -\epsilon); c + \epsilon]$,
  where $\epsilon$ is fixed to $0.1\;\si{\litre/\second}$.

\item \emph{The Pump.} The pump is either \on\ or \off, and we assume
  it is initially \off\ at the beginning of a cycle.  While it is~\on,
  it~pumps oil into the accumulator with a rate~$p_r =
  2.2\;\si{\litre/\second}$.  The~pump is also subject
  to timing constraints, which prevent switching it on and off too often. 
%
\item \emph{The Accumulator.} The volume of oil within the accumulator
  will be modelled by means of an energy variable~$v$.
%
  Its~evolution is given by the differential inclusion $dv/dt -u\cdot
  p_r\in -[m_r+\epsilon; m_r-\epsilon]$ (or~$-[m_r+\epsilon;0]$ if
  $m_r-\epsilon<0$), where $u\in\{0,1\}$ is the state of the pump.
\end{trivlist}
The controller must operate the pump (switch it on and~off) to ensure
the following requirements:
\begin{enumerate*}[label=(R\arabic*)]
\item the level of oil in the accumulator must always stay within the
  safety bounds $E = [V_{\min}; V_{\max}]$\footnote{The HYDAC company
    has fixed $V_{\min} = 4.9\litr$ \si{\litre} and $V_{\max} = 25.1\litr$.}
\item at the end of each machine cycle, the level of oil in the
  accumulator must ensure the controllability of the following cycle.
\end{enumerate*}

\smallskip
By modelling the oil pump system as a SETA~$\mathcal{H}$, the above
control problem can be reduced to finding a deterministic schedule
that results in a safe infinite run in~$\mathcal{H}$. Furthermore, we
are also interested in determining the minimal safety interval~$E$,
i.e., finding interval bounds that minimise~$V_{\max} - V_{\min}$,
while ensuring the existence of a valid controller for~$\mathcal{H}$.

As a first step in the definition of~$\calH$, we~build an ETP
representing the behaviour of the machine, depicted on Fig.~\ref{fig-setaconsum-app}.
\begin{figure}[ht]
  \centering
  \begin{tikzpicture}[xscale=1.3]
    \begin{scope}
\everymath{\scriptstyle}
\draw (0,0) node[rond5,orange] (a) {} node {$0$} node[below=3mm] {$x\leq 2$};
\draw (1,0) node[rond5,orange] (b) {} node {$-1.2$} node[below=3mm] {$x\leq 2$};
\draw (2,0) node[rond5,orange] (c) {} node {$0$} node[below=3mm] {$x\leq 2$};
\draw (3,0) node[rond5,orange] (d) {} node {$0$} node[below=3mm] {$x\leq 2$};
\draw (4,0) node[rond5,orange] (e) {} node {$-1.2$} node[below=3mm] {$x\leq 2$};
\draw (5,0) node[rond5,orange] (f) {} node {$-2.5$} node[below=3mm] {$x\leq 2$};
\draw (6,0) node[rond5,orange] (g) {} node {$0$} node[below=3mm] {$x\leq 2$};
\draw (7,0) node[rond5,orange] (h) {} node {$-1.7$} node[below=3mm] {$x\leq 2$};
\draw (8,0) node[rond5,orange] (i) {} node {$-0.5$} node[below=3mm] {$x\leq 2$};
\draw (9,0) node[rond5,orange] (j) {} node {$0$} node[below=3mm] {$x\leq 2$};
\draw[-latex'] (a) -- (b) node[midway,above] {$x=2$} node[midway,below] {$x:=0$};
\draw[-latex'] (b) -- (c) node[midway,above] {$x=2$} node[midway,below] {$x:=0$};
\draw[-latex'] (c) -- (d) node[midway,above] {$x=2$} node[midway,below] {$x:=0$};
\draw[-latex'] (d) -- (e) node[midway,above] {$x=2$} node[midway,below] {$x:=0$};
\draw[-latex'] (e) -- (f) node[midway,above] {$x=2$} node[midway,below] {$x:=0$};
\draw[-latex'] (f) -- (g) node[midway,above] {$x=2$} node[midway,below] {$x:=0$};
\draw[-latex'] (g) -- (h) node[midway,above] {$x=2$} node[midway,below] {$x:=0$};
\draw[-latex'] (h) -- (i) node[midway,above] {$x=2$} node[midway,below] {$x:=0$};
\draw[-latex'] (i) -- (j) node[midway,above] {$x=2$} node[midway,below] {$x:=0$};
    \end{scope}
  \end{tikzpicture}
  \caption{The ETP representing the oil consumption of the machine.}
  \label{fig-setaconsum-app}
  \medskip

  \begin{tikzpicture}
    \everymath{\scriptstyle}
    \draw (0,0) node[rond5,orange,opacity=.9,dashed] (a) {} node {$-m$} node[below=3mm] {$x\leq 2$};
    \draw (2,0) node[rounded corners=2mm,minimum width=1cm,minimum height=5mm,rouge] (b) {$p-m$} node[below=3mm] {$x\leq 2$};
    \draw (4,0) node[rounded corners=2mm,minimum width=1cm,minimum height=5mm,vert] (c) {$-m$} node[below=3mm] {$x\leq 2$};
    \draw (6,0) node[rond5,jaune,opacity=.9,dashed] (d) {$-m'$} node[below=3mm] {$x\leq 2$};
    \draw[-latex'] (a) -- (b);
    \draw[-latex'] (b) -- (c);
    \draw[-latex'] (c) -- (d) node[midway,above] {$x=2$} node[midway,below] {$x:=0$};
  \end{tikzpicture}
  \caption{An ETP for modelling the pump}
  \label{fig-setapump-app}
\end{figure}
In~order to fully model the behaviour of our oil-pump system, one
would require the parallel composition of this ETP with another
ETP representing the pump. The~resulting ETA would not be a flat SETA,
and would not fit in our setting. Since it~still provides interesting
result, we~develop this approach in~Appendix~\ref{app:hydacfull}).

Instead, we~consider a simplified model of the pump, which only allows
to switch it on and off once during each 2-second slot. This~is
modelled by inserting, between any two states of the model of
Fig.~\ref{fig-setaconsum-app}, a~copy of the ETP depicted on
Fig.~\ref{fig-setapump-app}. In~that ETP, the state with rate~$p-m$
models the situation when the pump is~on. Keeping the pump off for the
whole slot can be achieved by spending delay~zero in that state.
%
%
We~name~$\calH_1=\tuple{M,T,P_1}$ the SETA made of a single
macro-state equiped with a self-loop labelled with the ETP above.

In~order to take into account the timing constraints of the pump
switches, we~also consider a second SETA
model~$\calH_2=\tuple{M,T,P_2}$ where the pump can be operated only
during every other time slot. This amount to inserting the ETP of
Fig.~\ref{fig-setapump-app} only after the first, third, fifth, seventh
and ninth states of the ETP of Fig.~\ref{fig-setaconsum-app}.


We~also consider extensions of both models with uncertainty
$\epsilon=0.1\;\si{\liter/\second}$ (changing any negative rate~$-m$
into rate interval $[-m-\epsilon;-m+\epsilon]$, but changing rate~$0$
into $[-\epsilon;0]$). We~write~$\calH_1(\epsilon)$
and~$\calH_2(\epsilon)$ for the corresponding models.


For each model, we synthesise minimal upper bounds~$U$ (within the
interval~$[V_{\min};V_{\max}]$) that admit a solution to the
energy-constrained infinite-run problem for energy constraint $E =
[V_{\min};U]$. Then, we~compute the greatest stable interval~$[a;b] \subseteq
[L;U]$ of the cycle witnessing the existence of an $E$-constrained
infinite-run.  This~is done by closely following the methods described
in Sections~\ref{sec:ETA} and~\ref{sec:ETAu}.

Finally for each model we synthesise \emph{optimal} strategies that,
given an initial volume $w_0 \in [a,b]$ of the accumulator, return a
sequence of pump activation times~$t_i^{\text{on}}$
and~$t_i^{\text{off}}$ to be performed during the cycle.  This~is
performed in two steps: first, using Mjollnir, we~get a safe
\emph{permissive strategy} as a linear constraint linking~$w_0$,
the~intevals~$[L;U]$ and~$[a;b]$, and the times~$t_i^{\text{on}}$
and~$t_i^{\text{off}}$. We~then pick those safe delays that minimize
the average oil volume in the tank during one consumption cycle
(we~use the function \texttt{FindMinimum} of Mathematica to minimize
this non-linear function).
The resulting strategies are displayed on Fig.~\ref{fig:strategies-app}:
there, each horizontal line at a given initial oil level indicates the
delays (green intervals) where the pump will be running.


\begin{figure}[t]
\centering
\begin{subfigure}[b]{0.45\textwidth}
 \includegraphics[width=1 \textwidth]{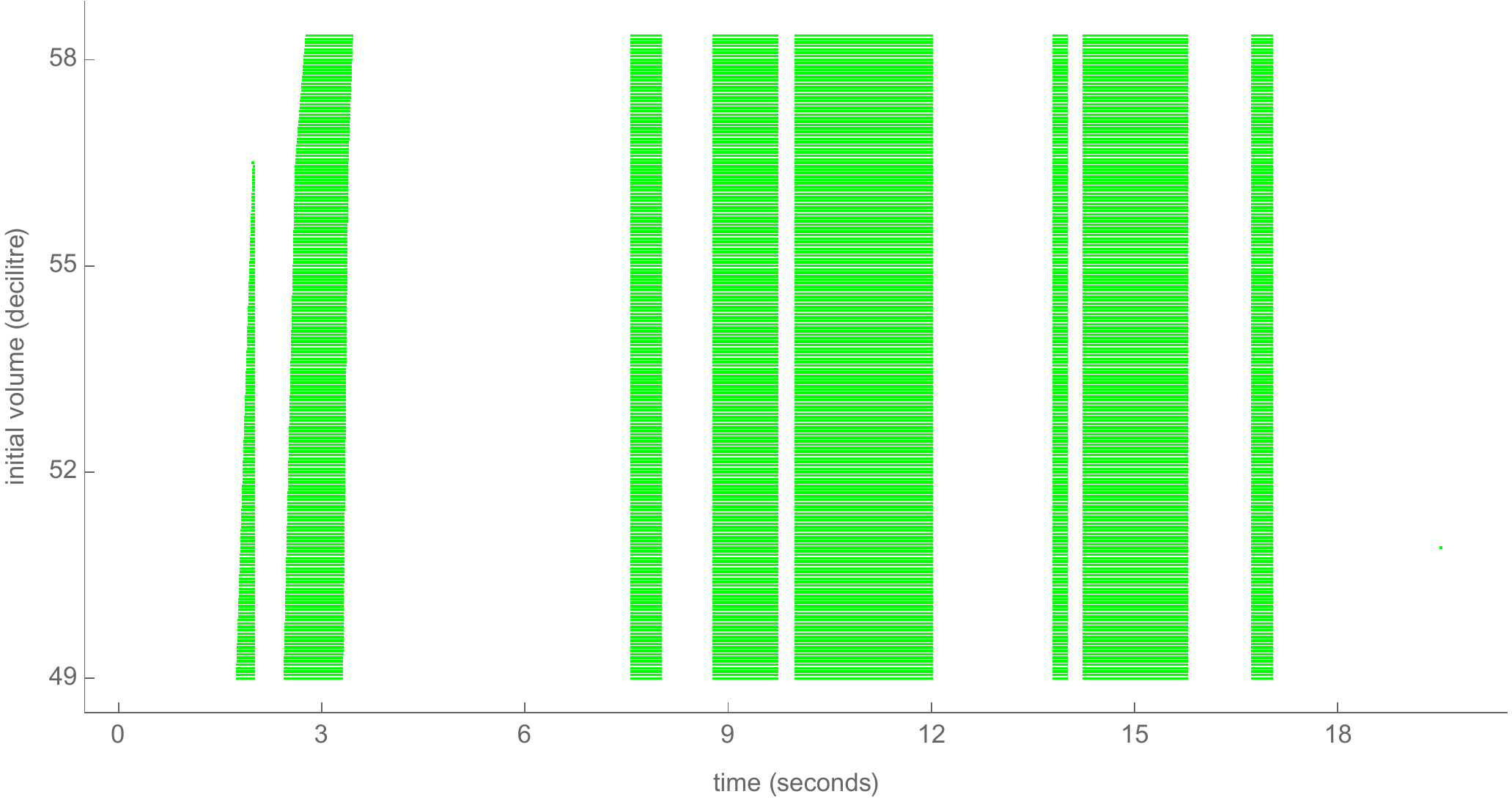}
\end{subfigure}
\quad
\begin{subfigure}[b]{0.45\textwidth}
 \includegraphics[width=1 \textwidth]{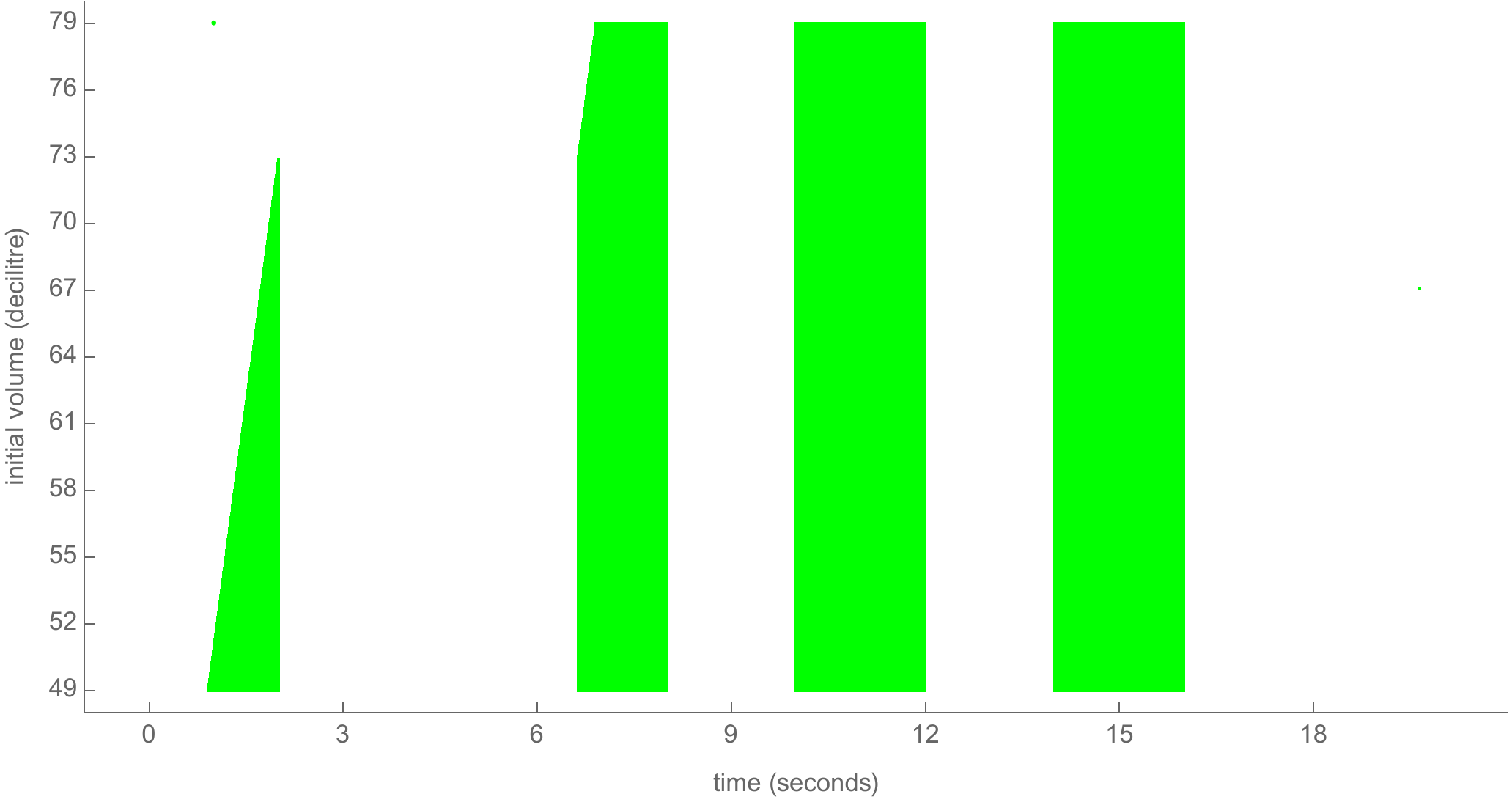}
\end{subfigure}
\\[1ex]
\begin{subfigure}[b]{0.45\textwidth}
 \includegraphics[width=1 \textwidth]{schedH1un}
\end{subfigure}
\quad
\begin{subfigure}[b]{0.45\textwidth}
 \includegraphics[width=1 \textwidth]{schedH2un}
\end{subfigure}
 \caption{Local strategies for a single cycle of the HYDAC system. (top-left) $\mathcal{H}_1$; (top-right) $\mathcal{H}_2$; (bottom-left) $\mathcal{H}_1(\epsilon)$; (bottom-right) $\mathcal{H}_2(\epsilon)$
 ($\epsilon = 0.1\;\si{\litre/\second}$).} \label{fig:strategies-app}
\end{figure}


\begin{table}[t]
\begin{center}
\footnotesize
\setlength{\tabcolsep}{1.5ex}
\begin{tabular}[t]{|c|c|c|c|}
\hline
Controller & $[L;U]$ & $[a;b]$ & Mean vol. (\si{\litre}) \\ \hline\hline
$\mathcal{H}_1$ & $[4.9;5.84]$ & $[4.9;5.84]$ & 5.43 \\ \hline
$\mathcal{H}_1(\epsilon)$ & $[4.9; 7.16]$ & $[5.1; 7.16]$ & 6.15 \\ \hline
$\mathcal{H}_2$ & $[4.9;7.9]$ & $[4.9;7.9]$ & 6.12 \\ \hline
$\mathcal{H}_2(\epsilon)$ & $[4.9;9.1]$ & $[5.1;9.1]$ & 7.24 \\ \hline 
\hline
G1M1~\cite{CJLRR09} & $[4.9;25.1]^{(*)}$ & $[5.1;9.4]$ & 8.2 \\ \hline
G2M1~\cite{CJLRR09} & $[4.9;25.1]^{(*)}$ & $[5.1;8.3]$ & 7.95 \\ \hline
\hline
\cite{DBLP:conf/fm/ZhaoZKL12} & $[4.9;25.1]^{(*)}$ & $[5.2;8.1]$ & 7.35 \\ \hline
\end{tabular}
\end{center}
\caption{Characteristics of the synthesised strategies, compared with the 
strategies proposed in~\cite{CJLRR09,DBLP:conf/fm/ZhaoZKL12} $(*)$ Safety interval given by the HYDAC company.}
\label{tab:schedulers-app}
\end{table}
The first part of Table~\ref{tab:schedulers-app} summarises the results
obtained for our models. It~gives the optimal volume constraints, the
greatest stable intervals, and the values of the worst-case (over all initial oil levels in~$[a;b]$)
mean volume.
It~is worth noting that the models without
uncertainty outperform the respective version with uncertainty.
Moreover, the~worst-case mean volume obtained both for
$\mathcal{H}_1(\epsilon)$ and $\mathcal{H}_2(\epsilon)$ are significantly
better than the optimal strategies synthesised both in~\cite{CJLRR09}
and~\cite{DBLP:conf/fm/ZhaoZKL12}.

The reason for this may be that
\begin{enumerate*}[label=(\roman*)]
\item our models relax the latency requirement for the pump, 
\item the strategies of~\cite{CJLRR09} are obtained using a
  discretisation of the dynamics within the system, and
\item the strategies of~\cite{CJLRR09}
  and~\cite{DBLP:conf/fm/ZhaoZKL12} where allowed to activate the pump
  respectively two and three times during each cycle.
\end{enumerate*}


\begin{figure}[t]
\centering
\begin{subfigure}[b]{0.45\textwidth}
 \includegraphics[width=1 \textwidth, height=4cm]{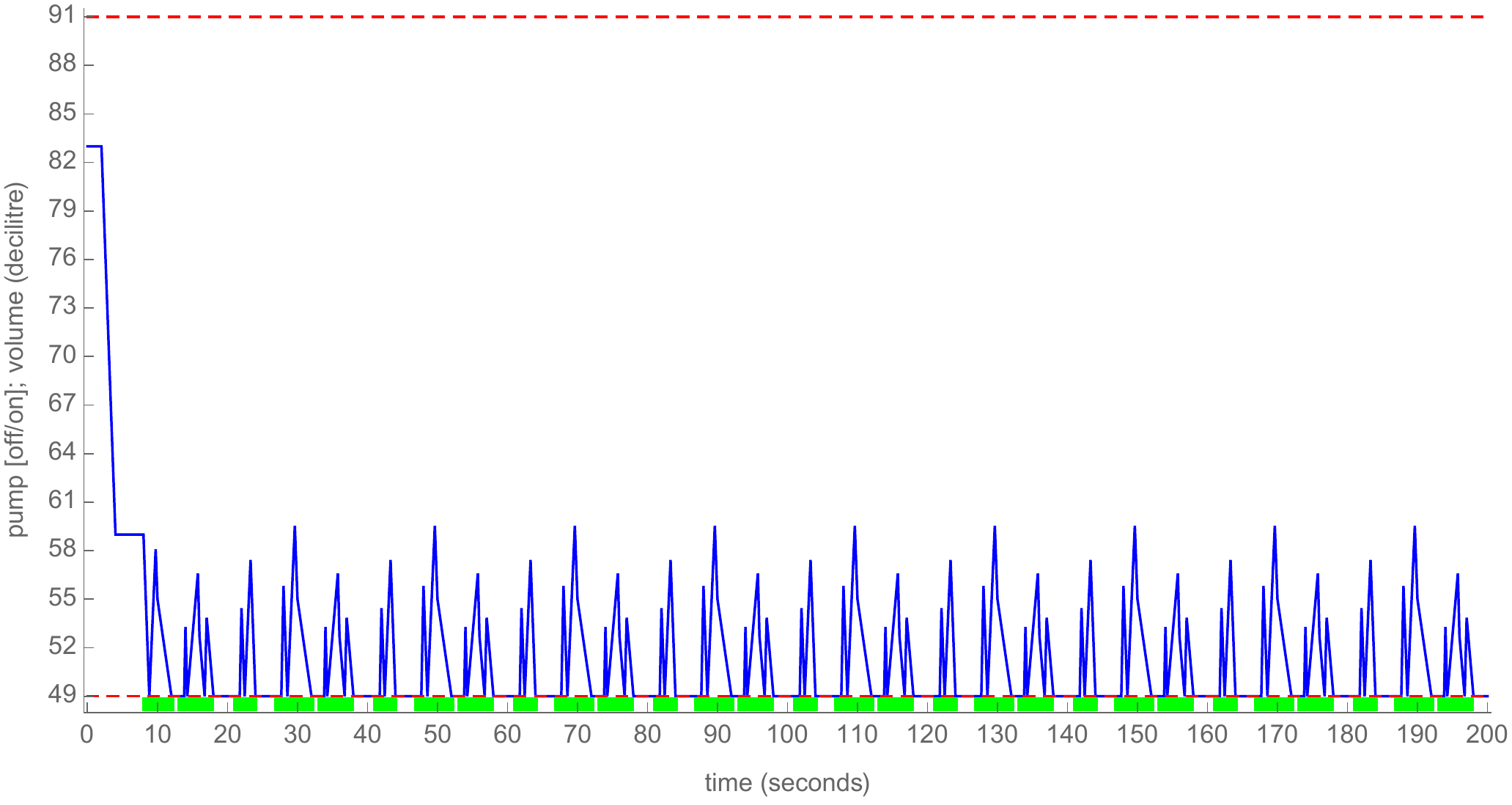}
\end{subfigure}
\quad
\begin{subfigure}[b]{0.45\textwidth}
 \includegraphics[width=1 \textwidth, height=4cm]{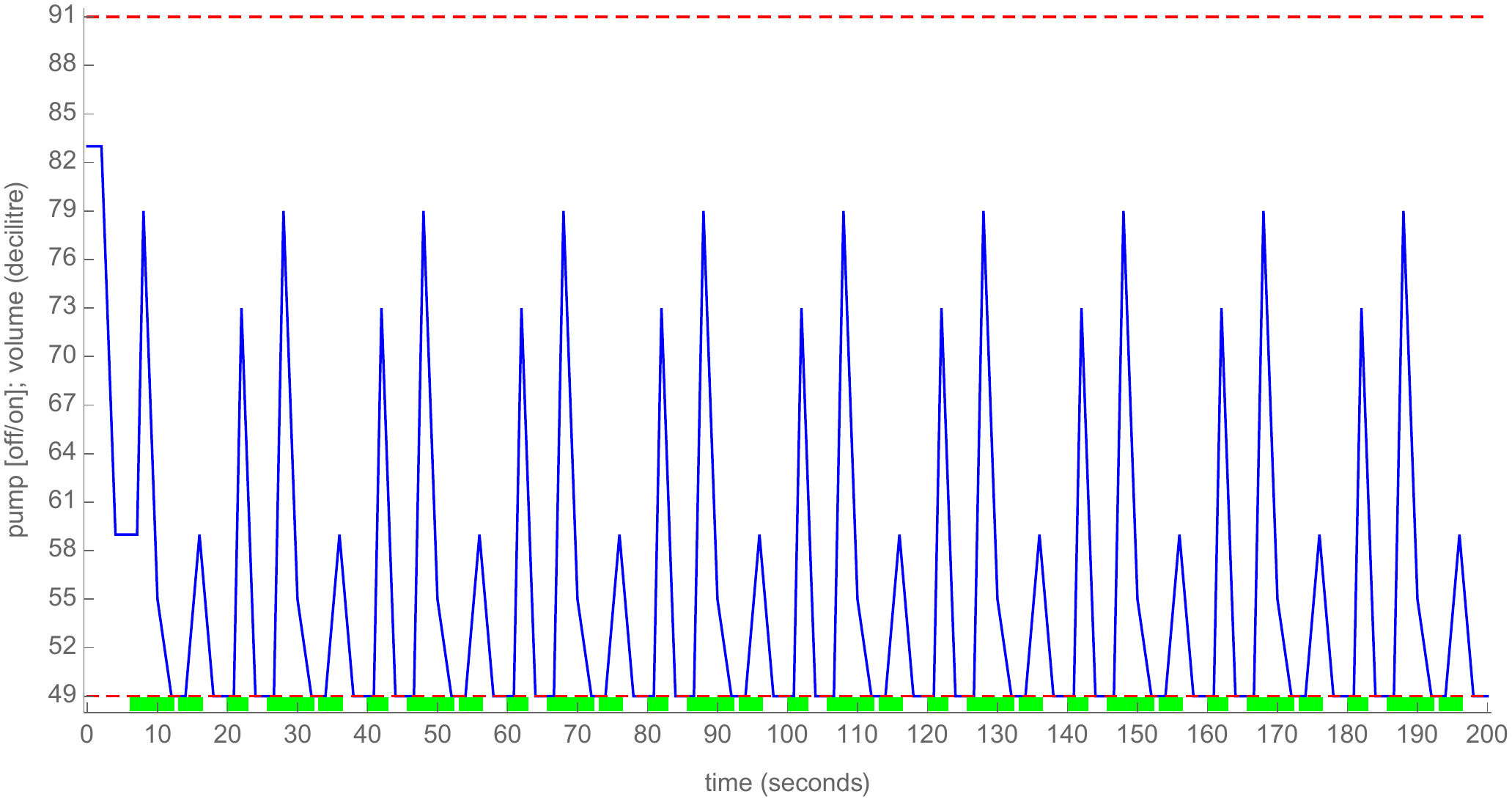}
\end{subfigure}
\\[1ex]
\begin{subfigure}[b]{0.45\textwidth}
 \includegraphics[width=1 \textwidth, height=4cm]{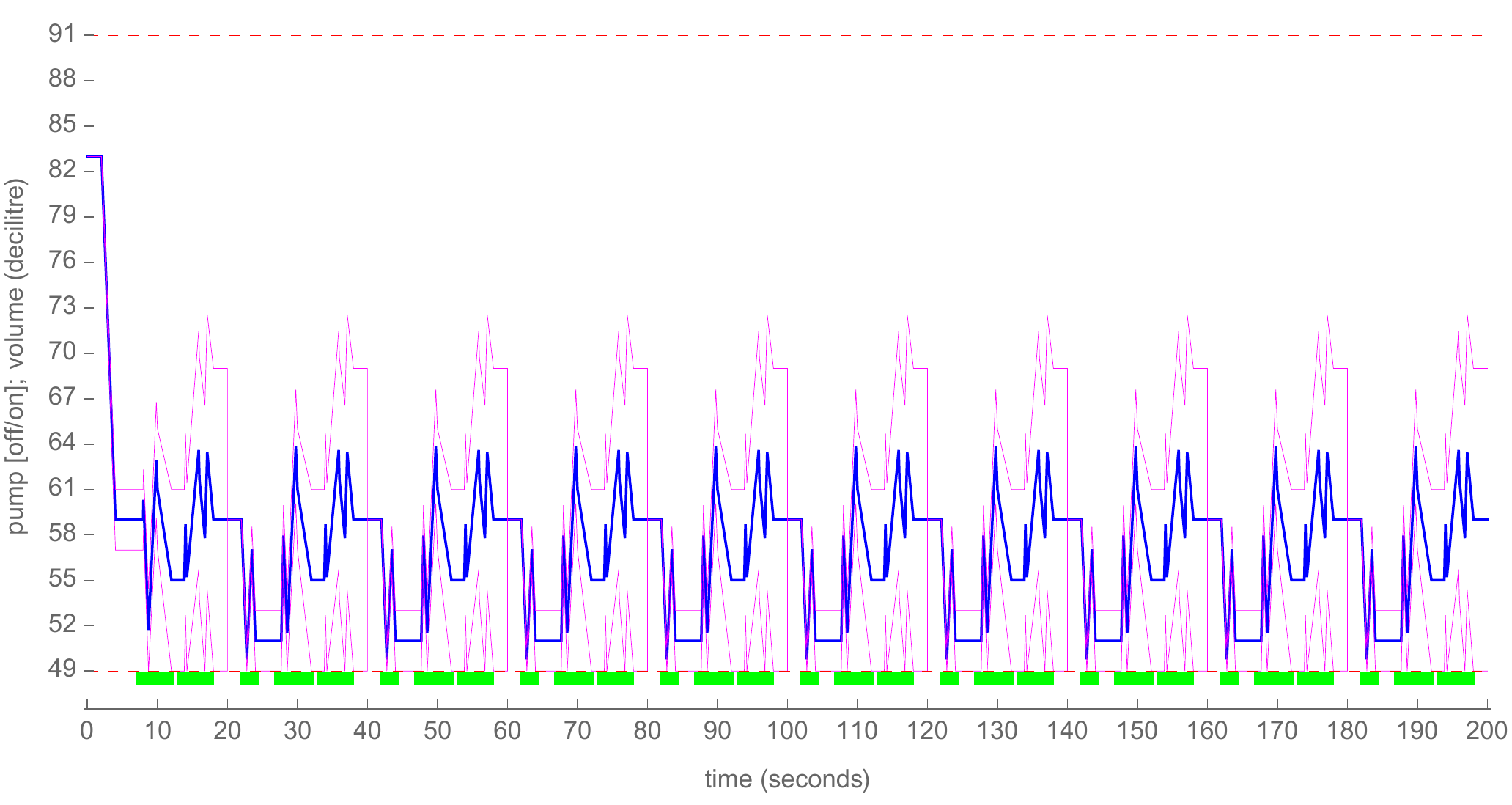}
\end{subfigure}
\quad
\begin{subfigure}[b]{0.45\textwidth}
 \includegraphics[width=1 \textwidth, height=4cm]{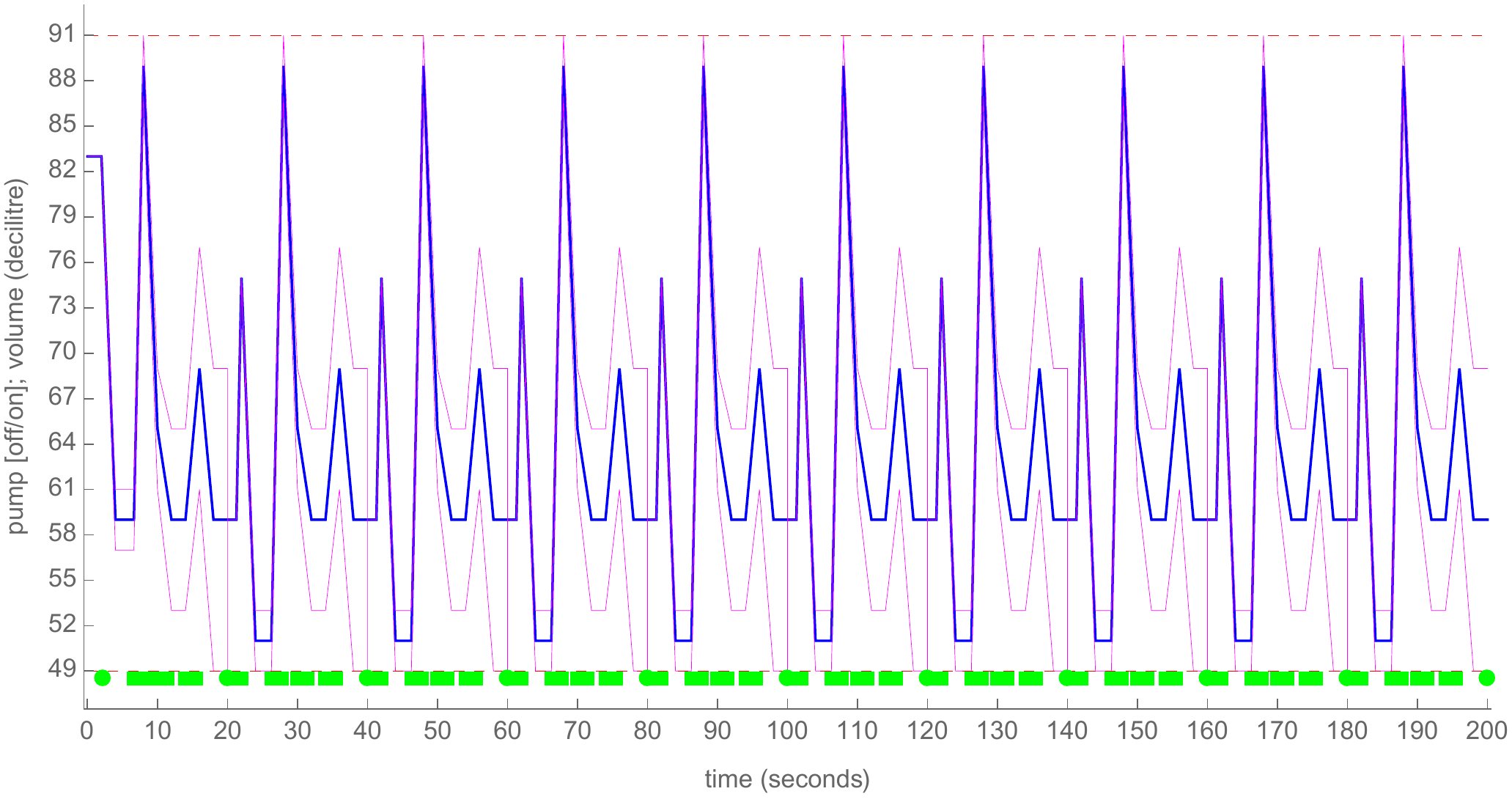}
\end{subfigure}
 \caption{Simulations of 10 consecutive machine cycles performed resp.\ with the strategies for (top-left) $\mathcal{H}_1$; (top-right) $\mathcal{H}_2$; (bottom-left) $\mathcal{H}_1(\epsilon)$; and (bottom-right) $\mathcal{H}_2(\epsilon)$.} \label{fig:simulations-app}
\end{figure}

We proceed by comparing the performances of our strategies in terms of
accumulated oil volume.  Figure~\ref{fig:simulations-app} shows the
result of simulating our strategies for a duration of~$200\;\si{\second}$,
i.e., $10$~machine cycles. The~plots illustrate the
dynamics of the oil level in the accumulator as well as the state
of the pump. The~initial volume used for evaluating the strategies is
$8.3\litr$, as done in~\cite{CJLRR09} for evaluating respectively the
Bang-Bang controller, the Smart Controller developed by HYDAC, and the
controllers G1M1 and G2M1 synthesised with
\textsc{uppaal-tiga}\footnote{We refer the reader to~\cite{CJLRR09}
  for a more detailed description of the controllers.}.

\begin{table}[t]
\begin{center}
\setlength{\tabcolsep}{1ex}
\begin{tabular}[t]{|c|c|c|}
\hline
Controller & Acc.\ vol. (\si{\litre}) & Mean vol. (\si{\litre}) \\ \hline
\hline
$\mathcal{H}_1$ & 1081.77 & 5.41 \\ \hline
$\mathcal{H}_2$ & 1158.9 & 5.79 \\ \hline
$\mathcal{H}_1(\epsilon)$ & 1200.21 & 6.00 \\ \hline
$\mathcal{H}_2(\epsilon)$ & 1323.42 & 6.62 \\ \hline 
\end{tabular}
\begin{tabular}[t]{|c|c|c|}
\hline
Controller & Acc.\ vol. (\si{\litre}) & Mean vol. (\si{\litre}) \\ \hline
\hline
Bang-Bang & 2689 & 13.45 \\ \hline
\textsc{hydac} & 2232 & 11.6 \\ \hline
G1M1 & 1518 & 7.59 \\ \hline
G2M1 & 1489 & 7.44 \\ \hline
\end{tabular}
\end{center}
\caption{Performance based on simulations of $200\;\si{\second}$ starting with $8.3\litr$.}
\label{tab:simulations-app}
\end{table}

Table~\ref{tab:simulations-app} presents, for each of the strategies, the
resulting accumulated volume of oil, and the corresponding mean
volume.  There is a clear evidence that the strategies for
  $\mathcal{H}_1$ and~$\mathcal{H}_2$ outperform all the other
  strategies. Clearly, this is due to the fact that they assume full
  precision in the rates, and allow for more switches of the pump.
  However, these results shall be read as what
  one could achieve by investing in more precise
  equipment.
%
The~results also confirm that both our strategies outperform those
presented in~\cite{CJLRR09}. In particular the strategy for
$\mathcal{H}_1(\epsilon)$ provides an improvement of $55\%$, $46\%$,
$20\%$, and $19\%$ respectively for the Bang-Bang controller, the
Smart Controller of HYDAC, and the two strategies synthesised with
\textsc{uppaal-tiga}.

\paragraph{Tool Chain.}
Our results have been obtained using Mathematica~\cite{Mathematica} and
Mjollnir~\cite{Monniaux10}. Specifically, Mathematica was used to
construct the formulas modelling the post fixed-points of the energy
functions while Mjollnir was used for performing quantifier
elimination on them.  The computation of the optimal upper bounds, and
greatest stable intervals were then handled with Mathematica, as well
as the computation of the optimal schedules and the respective
simulations.  It~is worth mentioning that Mathematica provides the
built-in function \texttt{Resolve} for preforming quantifier
elimination, but Mjollnir was preferred to it both for its
performances and its concise output.  By~calling Mjollnir from
Mathematica \emph{while constructing our predicates}, we~were able to
simplify formulas with more than $27$ quantifiers in approximately
$0.023$~sec. In contrast, resolving the same formula directly in
Mjollnir took us more that 20 minutes!

\input{hydac2}

%% file: hydac2.tex
\section{Non-flat model of the HYDAC case}\label{app:hydacfull}

We~briefly present a more precise model of the HYDAC example, closer
to what appeard in~\cite{CJLRR09}, using a non-flat SETA. The~model is
built by considering two flat ETPs running in parallel: one~ETP models
the consumption cycle of the machine (with fixed delays;
see~Fig.~\ref{fig-setaconsum}), and the second one models the state of
the pump over a complete cycle of the machine, allowing for instance
at most 4 switches during one cycle (see~Fig.~\ref{fig-hnonflat}). This
almost exactly corresponds to the model considered in~\cite{CJLRR09}.

\begin{figure}[ht]
\centering
  \begin{tikzpicture}[xscale=1.5]
    \begin{scope}[yshift=-1.6cm,xshift=1cm,xscale=1.2]
  \draw (0,0) node[rond6,bleu] (a) {$0$};
  \draw (1,0) node[rond6,bleu] (b) {$2.2$};
  \draw (2,0) node[rond6,bleu] (c) {$0$};
  \draw (3,0) node[rond6,bleu] (d) {$2.2$};
  \draw (4,0) node[rond6,bleu] (e) {$0$} node[below=3mm] {$t\leq 20$};
\everymath{\scriptstyle}
  \draw[-latex'] (a) -- (b) node[midway,above] {$y\geq 2$} node[midway,below] {$y:=0$};
  \draw[-latex'] (b) -- (c) node[midway,above] {$y\geq 2$} node[midway,below] {$y:=0$};
  \draw[-latex'] (c) -- (d) node[midway,above] {$y\geq 2$} node[midway,below] {$y:=0$};
  \draw[-latex'] (d) -- (e) node[midway,above] {$y\geq 2$} node[midway,below] {$y:=0$};
  \draw[-latex'] (e) -- +(0:10mm) node[midway,above] {$t=20$};
\end{scope}
\end{tikzpicture}
\caption{An ETP modelling the pump}
\label{fig-hnonflat}
\end{figure}
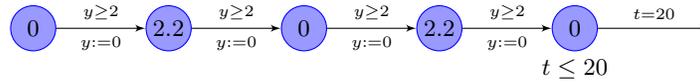

The resulting model is an ETA, which can actually be turned into a
non-flat SETA.
Hence it only fits in
our framework with uncertainty.  However, for~fixed~$L$ and~$U$, it~is
still possible to write down the energy relation, with or without
uncertainty: it~results in a (large) list of cases, because of
interleavings.

Following~\cite{CJLRR09}, we~then compute $m$-stable intervals, i.e.,
intervals~$[a;b]$ of oil levels for which there is a schedule to end
up with final oil level in~$[a+m;b-m]$.
In~the absence of uncertainties, fixing $L=4.9\litr$ and
$m=0.4\litr$, we~could then prove that there
are $m$-stable intervals as soon as $U\geq 8.1\litr$.


With~uncertainties, we~obtain an $m$-stable interval $[5.1;8.9]\litr$
as soon as $U\geq 11.5\litr$.  This again significantly improves
on~\cite{CJLRR09} (which considered discrete time).  Notice we did not
apply our algorithm based on Formula~\eqref{Ueq}~here (hence we~may
have missed better solutions): the~formula would be very large,
and~would involve $(U-L)/0.2$ intervals~$[a_{s,j};b_{s,j}]$.

For the $m$-stable interval~$[5.1;8.9]\litr$, we~computed the
constraints characterising all safe
strategies. Figure~\ref{fig-strat115} displays our strategies (notice
the similarities with Fig.~5 of~\cite{CJLRR09}). We~were not able to
select the optimal strategy for the mean volume because expressing the
mean volume results in a piecewise-quadratic function. Instead we
selected the strategy that fills in the tank as late as possible
(which intuitively tends to reduce the mean volume over one
cycle). A~simulation over 10 cycles is displayed on
Fig.~\ref{fig-simufull}.

  \begin{figure}[ht]
    \centering
    \begin{minipage}{.45\linewidth}
      \includegraphics[width=\linewidth] {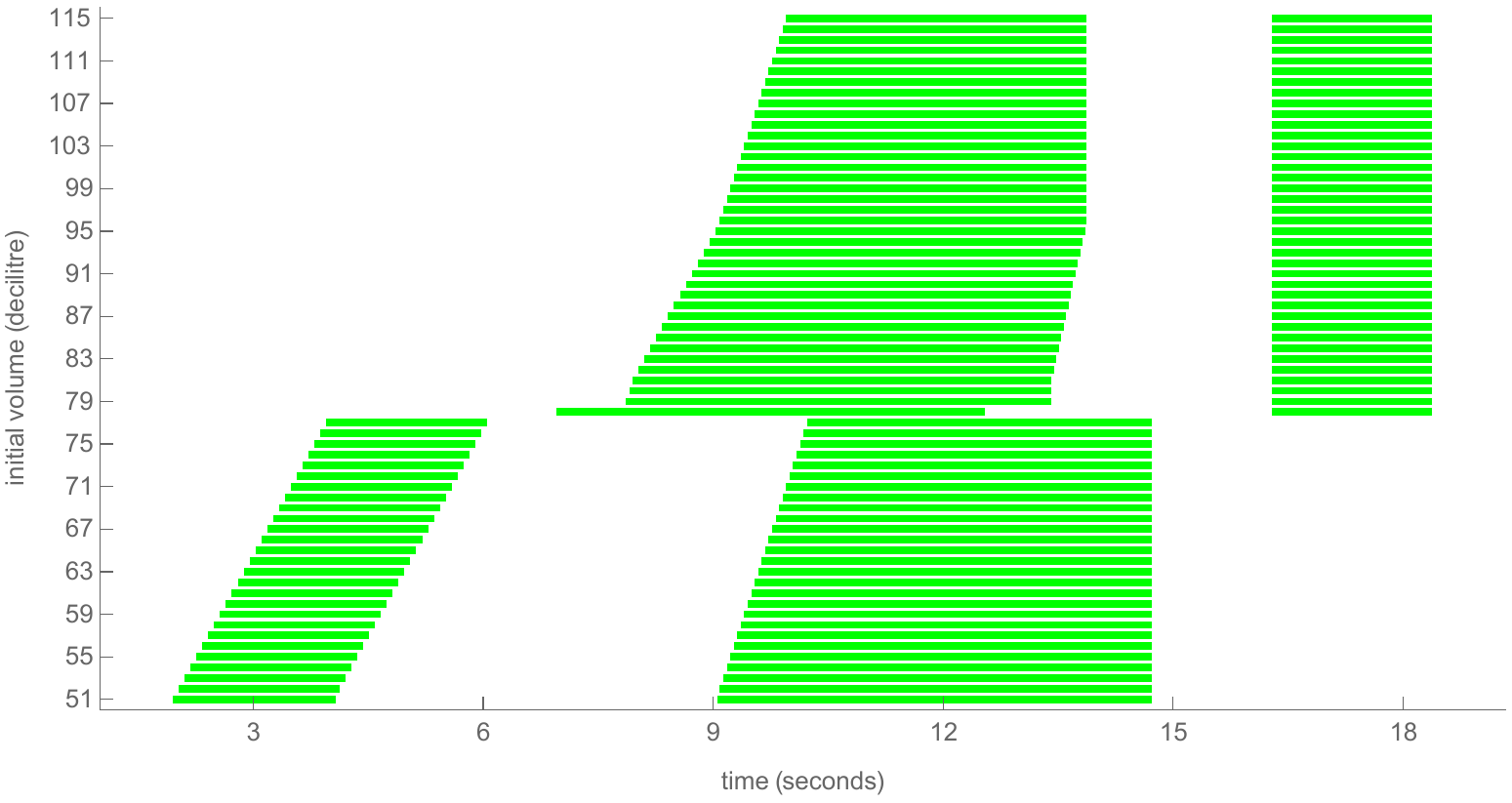}
      \caption{Strategies for the $m$-stable interval $[5.1;
          8.9]\litr$ (for $U=11.5\litr$)}\label{fig-strat115}
    \end{minipage}%
    \hfill
    \begin{minipage}{.45\linewidth}
      \includegraphics[width=\linewidth] {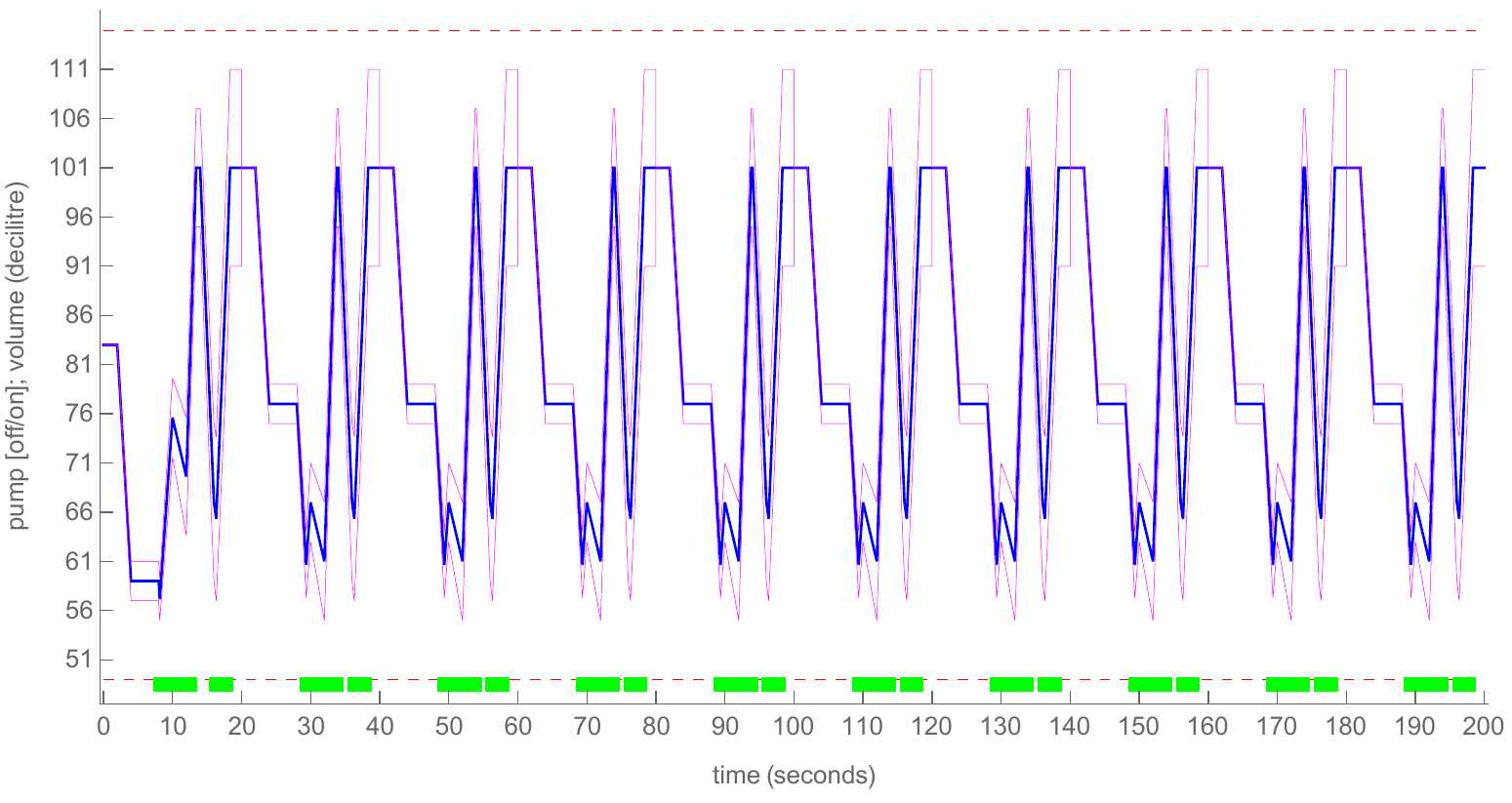}
      \caption{Simulation of 10 cycles}\label{fig-simufull}
    \end{minipage}
  \end{figure}